\documentclass[journal]{IEEEtran}

\usepackage{amsmath,graphicx,color,psfrag,mathrsfs,bm}
\usepackage{hyperref}

\usepackage{longtable}          % for multi-page tables
\usepackage{amsthm}
\usepackage{amssymb}

\usepackage{multicol}
\newtheorem{thm}{Theorem}[section]
\newtheorem{kor}[thm]{Corollar}
\newtheorem{lem}[thm]{Lemma}
\newtheorem{prop}[thm]{Proposition}

\theoremstyle{definition}

\theoremstyle{remark}
\newtheorem{rem}{Remark}[section]

\begin{document}
\title{Minimax Robust Hypothesis Testing}

\author{G\"{o}khan G\"{u}l,~\IEEEmembership{Student~Member,~IEEE,}
        Abdelhak M. Zoubir,~\IEEEmembership{Fellow,~EEE,}
\thanks{G. G\"{u}l and A. M. Zoubir are with the Signal Processing Group, Institute
of Telecommunications, Technische Universität Darmstadt, 64283, Darmstadt,
Germany (e-mail: ggul@spg.tu-darmstadt.de; zoubir@spg.tu-darmstadt.de)}% <-this % stops a space
\thanks{Manuscript received April 19, 2005; revised January 11, 2007.}}

\maketitle
\begin{abstract}
The minimax robust hypothesis testing problem for the case where the nominal probability distributions are subject to both modeling errors and outliers is studied in twofold. First, a robust hypothesis testing scheme based on a relative entropy distance is designed. This approach provides robustness with respect to modeling errors and is a generalization of a previous work proposed by Levy. Then, it is shown that this scheme can be combined with Huber's robust test through a composite uncertainty class, for which the existence of a saddle value condition is also proven. The composite version of the robust hypothesis testing scheme as well as the individual robust tests are extended to fixed sample size and sequential probability ratio tests. The composite model is shown to extend to robust estimation problems as well. Simulation results are provided to validate the proposed assertions.

\end{abstract}


% \fi
% \IEEEpeerreviewmaketitle

\section{Introduction}
\label{sec:intro}
The detection of the presence, or absence of an event with a specified accuracy is fundamental to statistical inference and binary hypothesis testing is the usual starting point. There are many applications, where binary hypothesis testing is used, for instance, radar, sonar, digital communications or seismology. A natural extension of binary hypothesis testing is multiple hypothesis testing, which builds a basis for classification and its importance is evident, for example, with pattern recognition. The necessity for statistical inference lies in the randomness that is inherent in the natural world such that received data, or signal, has an additive random component or, as in cognitive radio, must be modeled in a purely random manner. The degree of randomness in the received data usually turns out to be a metric of detection accuracy \cite{kay}.\\
Formally, any real world example of binary decision making problem can be modeled by a binary hypothesis test, where under each hypothesis $\mathcal{H}_j$, a received data $y\in\mathbb{R}$ follows a particular probability distribution $F_j$, $j\in\{0,1\}$. Accordingly, the aim is to find a decision rule $\delta$ which assigns each $y$ either to $\mathcal{H}_0$ or $\mathcal{H}_1$, depending on a certain objective function, which can be, for instance, the error probability. An optimal decision rule $\delta$ minimizes the objective function if $y$ indeed follows $F_j$ under $\mathcal{H}_j$, $j\in\{0,1\}$. However, this condition is too strict and often there are deviations from the model assumptions \cite{levy}.\\
A traditional way of considering the deviations from the nominal distributions is via parametric modeling.
Such parameters could be, for instance, the imprecisely known frequency of a receive signal or the unknown variance of a noise source.
The shape of the probability distributions under each hypothesis is still assumed to be completely known. However, this assumption is  invalid for various applications, for instance sonar, or cognitive radio.
Obviously, in such cases, a parametric model is inappropriate, or if such a model is used, then severe performance degradation results.\\
The shortcomings of parametric modeling necessitate the use of non-parametric approaches. Such approaches are robust, are cheap to implement in practice, make (almost) no assumption on the nominal distributions and their performance is acceptable for a variety of detection problems \cite{nonparametric}. However, compared to an optimum detector, their performance can be far away from being satisfactory, especially if there is some a priori knowledge available about the nominal distributions. Therefore, a more realistic approach should be tunable, depending on how much knowledge is available on the nominal distributions, how much robustness/performance trade-off is allowed as well as how complex the detector structure can be. In this context, robust minimax hypothesis testing falls between parametric and non-parametric detection; it coincides with parametric detection when the robustness parameters are chosen to be zero and it tends to a non-parametric test, a \it sign test \rm, when the robustness parameters are chosen to be at maximum \cite[p. 271]{hube81}.\\
A well known formulation of minimax hypothesis testing is based on building uncertainty sets $\mathcal{F}_j$ for each hypothesis $\mathcal{H}_j$, where $\mathcal{F}_j$ are populated by all probability distributions $G_j$, which are at least $\epsilon_j$ close to the nominal distribution $F_j$ with respect to some well defined distance $D_j$, $j\in\{0,1\}$. The choice of the parameters $\epsilon_0$ and $\epsilon_1$ determines the degree of robustness and they can vary with application. The eventual aim of the designer is to determine a pair of distributions $({G}_0, {G}_1)\in \mathcal{F}_0\times \mathcal{F}_1$, and a decision rule $\delta$, such that a predefined performance measure is met, e.g. the bounded error probability. This type of optimization is called minimax optimization and the distributions solving this problem are called least favorable distributions (LFD)s.\\
In this research field, there are two main approaches: one of which was initiated by Huber \cite{hube65} and the other by Dabak et al. \cite{dabak} and Levy \cite{levy09}. In Huber's work, which was published as early as $1965$, he proposed a robust version of the probability ratio test for $\epsilon-$contamination and total variation classes of distributions. He proved the existence of LFDs for both classes and showed that the resulting robust test was a censored/clipped version of the nominal likelihood ratio test. In a follow up work, he showed that the same conclusion could be made if the $\epsilon-$contamination model was extended to a larger class, which included five different distances as special cases \cite{hube68}. A more general uncertainty class, called $2$-alternating capacities, was proposed later by Huber and Strassen \cite{hube73}. However, it was noted in \cite{levy} that the approach in \cite{hube68} is more suitable for engineering applications due to its simplicity.
In a recent work, an uncertainty class which allows the use of composite distances for robust hypothesis testing has been proposed \cite{gul4}.\\
The robust tests pioneered by Huber were designed for modeling outliers. More recent works by Dabak \cite{dabak}, and later by Levy \cite{levy09}, show that when the distance $D$ is chosen to be the relative entropy, the resulting robust tests are different from Huber's robust test, depending on the choice of the objective function to be minimized. While Dabak's approach minimizes the relative entropy between the LFDs and provides an asymptotically robust test, Levy's robust test minimizes the type I and type II errors and provides a minimax robust test for a single sample. In \cite{levy}, it was noted that the latter two robust tests are more appropriate for modeling errors instead of outliers. Recently, it has been shown that Levy's robust test can be extended to distributed detection problems where the communication from sensors to the fusion center is constrained \cite{gul2}.
It has also been shown that considering the squared Hellinger distance instead of the relative entropy might provide a more flexible design \cite{gul}, \cite{gul3}.\\
In this paper, a robust hypothesis testing scheme based on Kullback-Leibler (KL) divergence is proposed. The problem formulation doesn't make any assumption about the choice of nominal distributions and, thus, it includes \cite{levy09} as a special case. This robust scheme is then extended by use of a composite uncertainty set, which is built with respect to two different distances.
The first distance models the misassumptions on the nominal distributions and the second distance models the outliers.
It is proven that LFDs for this composite model exist and therefore a single test can be robust with respect to both modeling errors as well as outliers. Notice that this composite class is different from the one proposed in \cite{gul4}. Finally, the designed robust tests are extended to fixed sample size and sequential probability ratio tests. It is also shown that the composite model can be extended to robust estimation problems.\\
The organization of this paper is as follows. In the following section, the LFDs and the robust decision rule are derived when the uncertainty sets are closed balls with respect to the KL divergence. The uniqueness and monotonicity properties of the LFDs are further proven. It is shown that the proposed model reduces to the model given in \cite{levy09}, when the nominal distributions are symmetric and the nominal likelihood ratio is monotone. For comparison reasons, the asymptotically robust test \cite{dabak} is presented and the existence of LFDs is proven without considering the geometrical aspects of hypothesis testing. The implications of considering other distances to obtain the LFDs and the robust decision rules are also discussed. In Sec.~\ref{section3}, the composite uncertainty set, which models both the outliers as well as the modeling errors, is introduced. When this model reduces to single robust tests, the density function of the log likelihood ratios is derived for performance evaluation as well
as for asymptotic analysis. Similarly, the equations from which one can uniquely determine the maximum of the robustness parameters, above which a minimax robust test cannot be designed, are also derived. In Sec.~\ref{section4}, the robust methods are extended to fixed sample size tests. Especially, it is shown whether the robust tests maintain their LFD properties. The section is concluded with obtaining the limiting tests and the formulation of asymptotic analysis. In Sec.~\ref{section5}, the sequential probability ratio test is robustified via replacing the nominal likelihood ratios by robust ones. It is investigated whether the LFD properties are preserved in general as well as asymptotically for both robustified sequential tests. In Sec.~\ref{section6}, an extension of the composite uncertainty model for the design of robust estimation problems is briefly introduced. In Sec.~\ref{simulations}, simulation results are presented and finally in Sec.~\ref{conclusions}, the paper is concluded. %\Longlefttarrow: Since $a=b$, we have $l(y_m-y)=l(y_m+y)^{-1}$ for all $y$
\section{Robust detection for modeling errors}
Let $(\Omega,{\mathscr{A}})$ be a measurable space with the probability measures $F_0$, $F_1$, $G_0$ and $G_1$ defined on it, which are absolutely continuous with respect to a dominating measure $\mu$, e.g. $\mu=F_0+F_1+G_0+G_1$. Furthermore, let $f_0$, $f_1$, $g_0$ and $g_1$ be the density functions of the probability measures $F_0$, $F_1$, $G_0$ and $G_1$ with respect to $\mu$, respectively. Define the uncertainty classes
\begin{equation}\label{eq7}
{\cal{G}}_j=\{g_j:D(g_j,f_j)\leq \varepsilon_j\}\quad j\in\{0,1\},
\end{equation}
where every $g_j$ is at least $\varepsilon_j>0$ close to the nominal density $f_j$, with respect to the KL-divergence i.e.
\begin{equation}\label{eq89}
D(g_j,f_j):=\int_{\mathbb{R}}\ln(g_j/f_j)g_j\mbox{d} \mu,\quad j\in\{0,1\}.
\end{equation}
Now, consider the composite hypothesis testing problem
\begin{align}\label{eq1}
\mathcal{H}_0&: Y \sim G_0 \nonumber\\
\mathcal{H}_1&: Y \sim G_1
\end{align}
where $Y$ is a real-valued random variable (r.v.) on $\Omega$. Define a randomized decision rule (function) $\delta\in\Delta$, where $\Delta$ stands for the set of all possible decision rules. Assume for the moment that $\varepsilon_0=\varepsilon_1=0$.
Then, the decision rule
\begin{equation}\label{equation23}
\delta(y) = \begin{cases} 0, &l(y)<\rho  \\ \kappa(y), & l(y)=\rho\\ 1, & l(y)> \rho \end{cases}
\end{equation}
for some threshold $\rho=P(\mathcal{H}_0)/P(\mathcal{H}_1)$ and a function $\kappa:\mathbb{R}\rightarrow [0,1]$, given the likelihood ratio $l(y):=f_1/f_0(y)$,
is optimum in the sense that it minimizes the error probability both in the Bayes and the Neyman-Pearson sense and results in two types of errors: the false alarm probability
\begin{equation}
P_E^0(\delta,f_0)=\int_{\mathbb{R}}\delta f_0\mbox{d} \mu\label{eq3}
\end{equation}
and the miss detection probability
\begin{equation}
P_E^1(\delta,f_1)=\int_{\mathbb{R}}(1-\delta)f_1\mbox{d} \mu.\label{eq4}
\end{equation}
Accordingly, the minimum error probability is given by
\begin{equation}
P_E(\delta,f_0,f_1)=P(\mathcal{H}_0)P_E^0(\delta,f_0)+P(\mathcal{H}_1)P_E^1(\delta,f_1).\label{eq5}
\end{equation}
\begin{rem}\label{rem1}
The sets ${\cal{G}}_0$ and ${\cal{G}}_1$ are not compact in the topology induced by the distance $D$. However, since $D$ is a convex function, ${\cal{G}}_0$ and ${\cal{G}}_1$ are convex sets. As a result ${\mathcal{G}}_0\times {\mathcal{G}}_1$ is also convex.
Given the a priori probabilities $P(\mathcal{H}_0)$ and $P(\mathcal{H}_1)$, the probability of error $P_E(\delta,f_0,f_1)$ is continuous, real-valued and linear, and therefore both convex and concave in all three terms $\delta,f_0,f_1$. In general, the space of all randomized decision rules $\Delta=C^0(\mathbb{R},[0,1])$ is not compact. The compactness condition, however, is not required because the error minimizing decision rules are known to \it exist \rm and to be the likelihood ratio test for all $(g_0,g_1)\in \mathcal{G}_0\times\mathcal{G}_1$. Let $\delta_1$ and $\delta_2$ be two decision functions chosen from $\Delta$. Then, simply for $\delta=\alpha\delta_1+(1-\alpha)\delta_2$, $0 \leq\alpha\leq 1$, we have $\delta\in\Delta$ and therefore $\Delta$ is convex. Note that any finitely supported quantization of $g_0$ and $g_1$ makes both ${\mathcal{G}}_0\times {\mathcal{G}}_1$ and $\Delta$ compact with respect to the standard topology. This is a straightforward result of Heine-Borel theorem \cite{Rudin}.
%Accordingly, one can omit the technicalities for all practical purposes.
\end{rem}
Remark.~\ref{rem1} indicates that Sion's minimax theorem \cite{sion} is applicable,
\begin{equation}\label{eq8}
\sup_{(g_0, g_1)\in {\cal{G}}_0\times{\cal{G}}_1}\min_{\delta\in\Delta}P_E(\delta,g_0,g_1)=\min_{\delta\in\Delta}\sup_{(g_0, g_1)\in {\cal{G}}_0\times{\cal{G}}_1}P_E(\delta,g_0,g_1).
\end{equation}
Hence, $P_E(\delta,g_0,g_1)$ possesses a saddle-value on $\Delta\times ({\cal{G}}_0\times{{\cal{G}}_1})$ with the least favorable densities $(\hat{g}_0,\hat{g}_1)\in{\cal{G}}_0\times{{\cal{G}}_1}$ and the robust decision rule $\hat{\delta}\in\Delta$, i.e., $\{\hat{\delta},(\hat{g}_0,\hat{g}_1)\}$, resulting from Eq.~\eqref{eq8}. Consequently
\begin{equation}\label{eq9}
P_E(\delta,\hat{g}_0,\hat{g}_1)\geq P_E(\hat{\delta},\hat{g}_0,\hat{g}_1)\geq P_E(\hat{\delta},g_0,g_1).
\end{equation}
Since $P_E$ is distinct in $g_0$ and $g_1$, it follows that
\begin{align}\label{equation12}
&P_E^0(\hat{\delta},g_0)\leq P_E^0(\hat{\delta},\hat{g}_0),\nonumber\\
&P_E^1(\hat{\delta},g_1)\leq P_E^1(\hat{\delta},\hat{g}_1).
\end{align}
\begin{thm}\label{thm10}
Let $l_l$ and $l_u$ be two real numbers with $0<l_l\leq 1\leq l_u<\infty$. Then, for
\begin{align}\label{eq22aak}
z(l_l,l_u)=&\int_{l<l_l}f_1\mbox{d}\mu+\int_{l_l<l<l_u}\left(l_l^{-1}l\right)^{\frac{\ln\left(k(l_l,l_u)\right)}{\ln\left(l_u/l_l\right)}}f_1\mbox{d}\mu\nonumber\\
+&k(l_l,l_u)\int_{l>l_u}f_1\mbox{d}\mu.
\end{align}
and
\begin{equation}\label{eq22aab}
k(l_l,l_u)=\frac{\int_{l<l_l}(l_l-l) f_0\mbox{d}\mu}{ \int_{l>l_u}(l-l_u)f_0\mbox{d}\mu}
\end{equation}
the least favorable densities
\begin{align}\label{equation24}
\hat{g}_0(y) &= \begin{cases} \frac{l_l}{z(l_l,l_u)}f_0(y), &l(y)<l_l  \\
\frac{1}{z(l_l,l_u)}\left(l_l^{-1}l(y)\right)^{\frac{\ln\left(k(l_l,l_u)\right)}{\ln\left(l_u/l_l\right)}}f_1(y), & l_l \leq l(y)\leq l_u \\
\frac{l_u k(l_l,l_u)}{z(l_l,l_u)}f_0(y), & l(y)>l_u \end{cases}\nonumber\\
\hat{g}_1(y) &= \begin{cases} \frac{1}{z(l_l,l_u)}f_1(y), &l(y)<l_l  \\
\frac{1}{z(l_l,l_u)}\left(l_l^{-1}l(y)\right)^{\frac{\ln\left(k(l_l,l_u)\right)}{\ln\left(l_u/l_l\right)}}f_1(y), & l_l \leq l(y)\leq l_u \\
\frac{k(l_l,l_u)}{z(l_l,l_u)}f_1(y), & l(y)>l_u \end{cases}
\end{align}
and the decision rule
\begin{equation}\label{equation25}
\hat{\delta}(y) = \begin{cases} 0, &l(y)<l_l \\
\frac{\ln(l(y)/l_l)}{\ln(l_u/l_l)} & l_l \leq l(y)\leq l_u \\
1, & l(y)>l_u \end{cases}
\end{equation}
which is equivalent to the robust likelihood ratio
\begin{equation}\label{equation26}
\hat{l}(y)=l_l^{\hat{\delta}(y)-1}l_u^{-\hat{\delta}(y)}l(y)
\end{equation}
form the saddle value condition for Eq.~\eqref{eq8}. Furthermore the parameters $l_l$ and $l_u$ can be determined by solving
\begin{align}\label{equation29}
-&\ln(z(l_l,l_u))+\frac{1}{z(l_l,l_u)}\Bigg[l_l\ln(l_l)\int_{l<l_l}f_0\mbox{d}\mu\nonumber\\
+&\int_{l_l<l<l_u}\left(l_l^{-1}l\right)^{\frac{\ln\left(k(l_l,l_u)\right)}{\ln\left(l_u/l_l\right)}}\ln\left(l\cdot\left(l_l^{-1}l\right)^{\frac{\ln\left(k(l_l,l_u)\right)}{\ln\left(l_u/l_l\right)}}\right)f_1\mbox{d}\mu\nonumber\\
+&k(l_l,l_u)l_u\ln(k(l_l,l_u)l_u)\int_{l>l_u}f_0\mbox{d}\mu\Bigg]=\epsilon_0
\end{align}
and
\begin{align}\label{equation30}
-&\ln(z(l_l,l_u))+\frac{1}{z(l_l,l_u)}\Bigg[\nonumber\\
&\int_{l_l<l<l_u}\left(l_l^{-1}l\right)^{\frac{\ln\left(k(l_l,l_u)\right)}{\ln\left(l_u/l_l\right)}}\ln\left(\left(l_l^{-1}l\right)^{\frac{\ln\left(k(l_l,l_u)\right)}{\ln\left(l_u/l_l\right)}}\right)f_1\mbox{d}\mu\nonumber\\
+&k(l_l,l_u)\ln(k(l_l,l_u))\int_{l>l_u}f_1\mbox{d}\mu\Bigg]=\epsilon_1.
\end{align}
\end{thm}
\begin{proof}
The solution of the minimax non-linear optimization problem
 \begin{align}\label{equation31}
\max_{g_j\in{\cal{G}}_j} \quad & P_E^{j}(\delta,g_j),\,\, j\in\{0,1\}\nonumber\\
\mbox{s.t.} \quad & g_j>0\nonumber\\
& \Upsilon(g_j)=\int_{\mathbb{R}}g_j\mbox{d}\mu=1,\,\, j\in\{0,1\}\nonumber\\
\min_{\delta\in\Delta} \quad & P_E(\delta,\hat{g}_0,\hat{g}_1),
\end{align}
directly leads to the assertion. First, the maximization stage is solved by considering the Karush-Kuhn-Tucker (KKT) multipliers. The subsequent minimization and optimization stages complete the proof.

\subsection{Maximization stage}
Consider the Lagrangian
\begin{align}\label{eq14}
L^j(g_j,\lambda_j,\mu_j)&=P_E^j(\hat{\delta},g_j)+\lambda_j(\varepsilon_j-D(g_j|f_j))\nonumber\\
&+\mu_j(1-\Upsilon(g_j))),\,\,\,j\in\{0,1\}.
\end{align}
where $\mu_j$, $\lambda_j$ are the KKT multipliers which are imposed to satisfy the constraints. Since $L^j$ is concave in $g_j$, a globally optimum solution is guaranteed if the necessary KKT conditions are met \cite{ekeland}. Writing \eqref{eq14} explicitly for $P_E^0$, it follows that
\begin{equation}\label{eq15}
L^0(g_0,\lambda_0,\mu_0)=\int_{\mathbb{R}}\Big[\delta g_0+\lambda_0\epsilon_0-\lambda_0\ln\frac{g_0}{f_0} g_0+\mu_0-\mu_0g_0\Big]\mbox{d}\mu
\end{equation}
Imposing the first KKT condition (stationarity), through taking G$\hat{\mbox{a}}$teaux's derivative of Eq.~\eqref{eq15} in the direction of $\psi$, yields
\begin{equation}\label{eq16}
\int_{\mathbb{R}}[\delta-\lambda_0\ln\frac{g_0}{f_0}-\lambda_0-\mu_0]\psi\mbox{d}\mu,
\end{equation}
which implies
\begin{equation}\label{eq17}
\delta-\lambda_0\ln\frac{g_0}{f_0}-\lambda_0-\mu_0=0,
\end{equation}
since $\psi$ is an arbitrary function. Hence, $\hat{g}_0$, and in a similar way $\hat{g}_1$ by solving \eqref{eq14} for $P_E^1$, can be obtained. The results are
\begin{equation}\label{eq18}
\hat{g}_0=c_1 \left(\frac{c_2}{c_1}\right)^\delta f_0,\,\,\quad \hat{g}_1=c_3 \left(\frac{c_4}{c_3}\right)^\delta f_1
\end{equation}
where $c_1=\exp({-\frac{\lambda_0+\mu_0}{\lambda_0}})$, $c_2=\exp({-\frac{-1+\lambda_0+\mu_0}{\lambda_0}})$, $c_3=\exp({-\frac{-1+\lambda_1+\mu_1}{\lambda_1}})$ and $c_4=\exp({-\frac{\lambda_1+\mu_1}{\lambda_1}})$.
This leads to the robust likelihood ratio
\begin{equation}\label{eq19}
\hat{l}=\left(\frac{c_3}{c_1}\right) e^{-\delta\ln\left(\frac{c_2c_3}{c_1c_4}\right)} l.
\end{equation} %l_0^{\hat{\delta}(y)+\mu_0}l_1^{1-\hat{\delta}(y)-\mu_1}l(y)
%
%Let
%\begin{equation}\nonumber
%c_1=e^{-\frac{\lambda_0+\mu_0}{\lambda_0}}, \,\,c_2=e^{-\frac{-1+\lambda_0+\mu_0}{\lambda_0}}
%\end{equation}
%\begin{equation}\nonumber
%c_3=e^{-\frac{-1+\lambda_1+\mu_1}{\lambda_1}}, \,\,c_4=e^{-\frac{\lambda_1+\mu_1}{\lambda_1}}\,\,c_l=e^{-\frac{-1+\lambda_0+\lambda_1+\mu_0+\mu_1}{\lambda_0+\lambda_1}}
%\end{equation}
\subsection{Minimization stage}
The decision rule $\delta$, which minimizes $P_E$ for any $(g_0,g_1)\in\mathcal{G}_0\times\mathcal{G}_1$, is known to be the likelihood ratio test \eqref{equation23}.
Solving $\hat{l}=1$ from Eq.~\eqref{eq19} and rewriting Eq.~\eqref{equation23} with $\rho=1$ for $\hat{l}$ yields
%Applying this to \eqref{eq19} for $\rho=1$, we get
\begin{equation}\label{eq211}
\hat{\delta}=\begin{cases} 0, &\hat{l}<1  \\
\frac{\ln\left(\frac{c_3}{c_1}l\right)}{\ln\left(\frac{c_2c_3}{c_1c_4}\right)}, & \hat{l}=1\\
 1, & \hat{l}> 1 \end{cases}
\end{equation}
Applying \eqref{eq211} to \eqref{eq18}, the least favorable distributions with respect to their density functions are obtained as
\begin{equation}\label{equation39}
\hspace{-2.5mm}\hat{g}_0 = \begin{cases} c_1 f_0, &\hat{l}<1 \\
c_0 l^{\frac{\ln(c_2/c_1)}{\ln\left(\frac{c_2 c_3}{c_1c_4}\right)}}f_0,& \hat{l}=1 \\
c_2 f_0, & \hat{l}> 1 \end{cases},\,\,
\hat{g}_1 = \begin{cases} c_3 f_1, &\hat{l}<1 \\
c_0 l^{\frac{\ln(c_4/c_3)}{\ln\left(\frac{c_2 c_3}{c_1c_4}\right)}}f_1,& \hat{l}=1 \\
c_4 f_1, & \hat{l}> 1 \end{cases}
\end{equation}
where $c_0=\exp\left({\frac{\ln c_2\ln c_3-\ln c_1 \ln c_4}{\ln(c_2 c_3)-\ln({c_1 c_4})}}\right)$. The unknown parameters can be obtained by imposing the constraints, or equivalently by solving the non-linear equations
\begin{align}\label{eq27}
&c_1\int_{\hat{l}<1} f_0\mbox{d}\mu+\int_{\hat{l}=1}\Phi\mbox{d}\mu+c_2\int_{\hat{l}> 1} f_0 \mbox{d}\mu=1\nonumber\\
&c_3\int_{\hat{l}<1} f_1\mbox{d}\mu+\int_{\hat{l}=1}\Phi\mbox{d}\mu+c_4\int_{\hat{l}> 1} f_1 \mbox{d}\mu=1\nonumber\\
&c_1\ln c_1\int_{\hat{l}<1} f_0\mbox{d}\mu+\int_{\hat{l}=1} \Phi \ln\frac{\Phi}{f_0} \mbox{d}\mu+c_2\ln c_2\int_{\hat{l}> 1} f_0\mbox{d}\mu=\varepsilon_0\nonumber\\
&c_3\ln c_3\int_{\hat{l}<1} f_1\mbox{d}\mu+\int_{\hat{l}=1} \Phi \ln\frac{\Phi}{f_1} \mbox{d}\mu+c_4\ln c_4\int_{\hat{l}> 1} f_1\mbox{d}\mu=\varepsilon_1
\end{align}
where $\Phi=c_0 \exp \left({\frac{(\ln c_4-\ln c_3)\ln l}{\ln({c_2 c_3})-\ln{(c_1c_4)}}}\right) f_1$. Note that the first two equations are required to make sure that $\hat{g}_0$ and $\hat{g}_1$ are density functions, i.e., they integrate to one and the other two equations are required to guarantee that $\hat{g}_0\in\mathcal{G}_0$ and $\hat{g}_1\in\mathcal{G}_1$.
\subsection{Optimization stage}
To complete the proof it is necessary to explain how $\hat{\delta}$, $\hat{g}_0$, $\hat{g}_1$ and the nonlinear equations can be represented in terms of $l_l$ and $l_u$. Let $l_l=c_1/c_3$ and $l_u=c_2/c_4$. Then, considering $\hat{l}=\hat{g}_1/\hat{g}_0$ from \eqref{equation39}, it follows that
 \begin{align}
&{\cal{R}}_1:=\{y:l(y)<{l_l}\} \equiv \{y:\hat{l}(y)<1\}\nonumber\\
&{\cal{R}}_2:=\{y:l_l\leq l(y)\leq l_u\} \equiv \{y:\hat{l}(y)=1\}\nonumber\\
&{\cal{R}}_3:=\{y:l(y)>{l_u}\} \equiv \{y:\hat{l}(y)>1\}\nonumber
\end{align}
Rewriting the integrals with the new limits (over $(\mathcal{R}_1,\mathcal{R}_2,\mathcal{R}_3)$), using the substitutions $c_1:=c_3 l_l$ and $c_2:=c_4 l_u$, dividing both sides of the first two equations in \eqref{eq27} by $c_3$, and equating them to each other via $1/c_3$ results in $c_4=k(l_l,l_u)c_3$. Accordingly, it follows that
\begin{equation}\label{eqsth}
\Phi=c_3 \left(l_l^{-1}l\right)^{\frac{\ln\left(k(l_l,l_u)\right)}{\ln\left(l_u/l_l\right)}}f_1.
\end{equation}
This allows the second equation in \eqref{eq27} to be written as $c_3:=1/z(l_l,l_u)$. Now, all constants $c_1,c_2,c_3$ and $c_4$ as well as $\Phi$ are parameterized by $l_l$ and $l_u$. Thus, Eq.~\eqref{eq19} can be rewritten as Eq.~\eqref{equation26} and $\hat{\delta},\hat{g}_0,\hat{g}_1$ as given in Theorem~\ref{thm10}. Finally, the last two equations of \eqref{eq27} reduce to \eqref{equation29} and \eqref{equation30}. This completes the proof.
\end{proof}
\subsection{Monotonicity of the relative entropy}
In the sequel it is shown that ordering in likelihood ratios implies ordering in KL-divergence. This explains the monotonic behavior of LFDs for increasing robustness parameters given that $l$ is monotone. The theory that will be presented will also be used in the next sections.
\begin{prop}\label{proppp2}
Let $F$ and $G$ be two probability measures on $(\Omega,{\mathscr{A}})$ with $\partial F/\partial G$ a non-decreasing function. Then, $G(y)\geq F(y)$ for all $y\in\mathbb{R}$.
\end{prop}
\begin{proof}
Due to a special case of the Fortuin-Kasteleyn-Ginibre (FKG) inequality, for any random variable $X$ and any two positive increasing functions $\phi$, $\psi$ we have $\mathrm{E} \, \phi(X) \psi(X) \ge \mathrm{E} \, \phi(X) \, \mathrm{E} \, \psi(X)$.
Applying this to $X$ distributed according to $G$ and the functions $\phi := \mathbf{1}_{[c,+\infty)}$, where $\mathbf{1}_{\{\cdot\}}$ is the indicator function, and $\psi := dF/dG$, we get $G(y)\geq F(y)$ for all $y\in\mathbb{R}$.
\end{proof}

\begin{rem}\label{rem2}
Let $X$ and $Y$ be two random variables defined on the same probability space $(\Omega,{\mathscr{A}})$, having continuous cumulative distribution functions $F$ and $G$, respectively. $X$ is called stochastically larger than $Y$, i.e., $X\succ_{ST}Y$, if $G(y)\geq F(y)$ for all $y$.
\end{rem}

\begin{kor}\label{korr1}
For every non-decreasing function $\phi$, $X\succ_{ST}Y\Longleftrightarrow\phi(X)\succ_{ST}\phi(Y)$, hence $X\succ_{ST}Y \Longleftrightarrow \mathrm{E}[\phi(X)]\geq \mathrm{E}[\phi(Y)]$
\end{kor}
Proof of Corollary \ref{korr1} is simple and can be found for example in \cite{Wolfstetter96}.
\begin{thm}\label{thm2}
Let $X_0$, $Y_0$, $X_1$, and $Y_1$ be four continuous random variables defined on $(\Omega,{\mathscr{A}})$ and having distinct densities $f_0$, $g_0$, $f_1$, and $g_1$, respectively, with $f_1/g_1$, $g_1/g_0$, and $g_0/f_0$, all being non-decreasing functions. Then,
\begin{equation}\label{eq45}
D(f_1,f_0)>D(g_1,g_0)\quad\mbox{and}\quad D(f_0,f_1)>D(g_0,g_1)
\end{equation}
\end{thm}
\begin{proof}
By Prop.~\ref{proppp2} and Remark~\ref{rem1}, we have $Y_1\succ_{ST}Y_0$ and $Y_0\succ_{ST}X_0$ since $g_1/g_0$, and $g_0/f_0$ are non-decreasing functions. Increasing $f_1/g_1$ and $g_1/g_0$ implies increasing $f_1/g_0$ and using Corollary \ref{korr1}, and denoting $\phi(Y)=\ln g_0/f_1(Y)$, we have $\mathrm{E}_{X_0}[\phi(Y)]\geq \mathrm{E}_{Y_0}[\phi(Y)]$. Hence, the identity $D(f_0,f_1)=\mathrm{E}_{X_0}[\phi(Y)]+D(f_0,g_0)$, together with $\mathrm{E}_{X_0}[\phi(Y)]\geq \mathrm{E}_{Y_0}[\phi(Y)]$, results in $D(f_0,f_1)\geq D(f_0,g_0)+D(g_0,f_1)\Longrightarrow D(f_0,f_1)>D(g_0,f_1)$. It is well known that
\begin{equation}\label{eq48}
\mathrm{E}_{Y_1}\left[\ln\frac{g_1(Y)}{f_1(Y)}\right]>0.
\end{equation}
Again, using the Corollary \ref{korr1}, and denoting $\psi(Y)=\ln f_1/g_1(Y)$, we have $\mathrm{E}_{Y_1}[\psi(Y)]\geq \mathrm{E}_{Y_0}[\psi(Y)]$, which implies $-\mathrm{E}_{Y_0}[\psi(Y)]>0$ in comparison with $-\mathrm{E}_{Y_1}[\psi(Y)]>0$.
We conclude that $D(f_0,f_1)>D(g_0,f_1)$ together with $-\mathrm{E}_{Y_0}[\psi(Y)]>0$ implies $D(f_0,f_1)>D(g_0,g_1)$. The proof for the case $D(f_1,f_0)>D(g_1,g_0)$ is similar and is omitted.
\end{proof}
Now, let the likelihood ratio with respect to the nominal distributions, $f_1/f_0$, be monotonically increasing. From Eq.~\eqref{equation24} it follows that $g_1/g_0$ and $f_1/g_1$ are all non-decreasing functions. Theorem~\ref{thm2} indicates that $D(f_i,f_{1-j})>D(g_i,g_{1-j})$ for $j\in\{0,1\}$, and this implies that $g_0$ and $g_1$ move towards each other monotonically.

\subsection{Symmetric density functions}
Depending on the extra constraints imposed on the nominal probability distributions, the equations that need to be
solved to determine the parameters of the LFDs can be simplified. Assume $f_0(y)=f_1(-y)$ for all $y\in\mathbb{R}$ and $\varepsilon=\varepsilon_0=\varepsilon_1$. This implies $l_u=1/l_l$. With this assumption Eq.~\eqref{equation29} and Eq.~\eqref{equation30} reduce to
\begin{align}\label{eq30}
-&\ln(z(l_u))-\frac{1}{z(l_u)}\Bigg(l_u^{-1/2}\ln l_u\int_{1<l<l_u}\sqrt{f_0f_1}\mbox{d}\mu\nonumber\\
+&l_u^{-1}\ln l_u\int_{l>l_u}f_1 \mbox{d}\mu\Bigg)=\varepsilon
\end{align}
where
\begin{align}\label{eq31}
z(l_u)&=\int_{l<1/l_u}f_1 \mbox{d}\mu+2 l_u^{-1/2}\int_{1<l<l_u}\sqrt{f_0f_1}\mbox{d}\mu\nonumber\\
&+l_u^{-1}\int_{l>l_u}f_1 \mbox{d}\mu.
\end{align}
The symmetry condition also implies $l(y)=1/l(-y)$ and $\hat{l}(y)=1/\hat{l}(-y)$ for all $y$. Accordingly, it follows that $k(l_l,l_u)={l_u}^{-1}$  and
$\hat{g}_0(y)=\hat{g}_1(-y)\forall y$. Notice that if $l$ is monotone, Eq.~\eqref{eq30} can be redefined in terms of $y_u$ by $l_u=l(y_u)$, $\{l>l_u\}\equiv(y_u,\infty)$ and due to symmetry $1/l_u=l(-y_u)$, $\{l<1/l_u\}\equiv(-\infty,-y_u)$. This proves that Theorem~\ref{thm10} is a generalization of the results of \cite{levy09}.
\subsection{Asymptotically robust hypothesis test}
So far, the problem of minimax robust hypothesis testing, for the case where the objective function to maximize was the error probability, has been studied. For the same uncertainty model \eqref{eq7}, Dabak and Johnson proposed a geometrically based robust detection scheme much earlier than \cite{levy09}. From \cite[p.254]{levy}, it is also known that the work of Dabak can be recreated by considering the same minimax optimization problem that has been introduced, see \eqref{equation31}, but changing the objective functions $P_E^0$ and $P_E^1$ to $-D(g_0,\bar{g}_1)$ and $-D(g_1,\bar{g}_0)$. Here, $D$ is again the relative entropy and $(\bar{g}_0,\bar{g}_1)$ are the least favorable densities,
\begin{equation}\label{equation61}
\bar{g}_0(y)=\frac{w(y;u)}{k(u)},\quad \bar{g}_1(y)=\frac{w(y;1-v)}{k(1-v)}
\end{equation}
where $u,v$ are parameters to be determined such that
\begin{equation}\label{equation62}
D(\bar{g}_0,f_0)=\epsilon_0,\quad D(\bar{g}_1,f_1)=\epsilon_1.
\end{equation}
Again by \cite{levy}, the fixed sample size test in the log domain
\begin{equation}\label{equation62q}
\frac{1}{n}\sum_{i=1}^ n \ln \left(l(y_i)\right)\stackrel{{\cal{H}}_1}{\underset{{\cal{H}}_0}{\gtrless}}\frac{\ln\left(\frac{k(1-v)}{k(u)}\right)}{1-(u+v)},
\end{equation}
is still a likelihood ratio test, but with a modified threshold ($\gamma=1$). The following proposition and the proof show that $\bar{g}_0$ and $\bar{g}_1$ are indeed LFDs without consideration of the geometrical aspects of hypothesis testing.
\begin{prop}
The pair of density functions $\bar{g}_0$ and $\bar{g}_1$ satisfy
\begin{equation}\label{equation62t}
\bar{g}_0=\max_{g_0\in\mathcal{G}_0} \mathrm{E}_{g_0}\ln\left(\frac{\bar{g}_1}{\bar{g}_0}\right),\quad \bar{g}_1=\min_{g_1\in\mathcal{G}_1} \mathrm{E}_{g_1}\ln\left(\frac{\bar{g}_1}{\bar{g}_0}\right).
\end{equation}
\end{prop}
\begin{proof}
Consider the Lagrangian function defined in \eqref{eq14}, where the objective functions $P_E^0$ and $P_E^1$ are replaced by $\mathrm{E}_{g_0}\ln\left({\bar{g}_1}/{\bar{g}_0}\right)$ and $\mathrm{E}_{g_1}\ln\left({\bar{g}_1}/{\bar{g}_0}\right)$ \eqref{equation62t}. Then, following similar steps to \eqref{eq15}-\eqref{eq18}, it can be shown that $g_0$ and $g_1$ have the same parametric forms as given in \eqref{equation61}. The equations in \eqref{equation62} are convex \cite{levy}, hence their solution is unique. Since $(g_0,g_1)$ must satisfy \eqref{equation62} with the same $(\epsilon_0,\epsilon_1)$ that $(\bar{g}_0,\bar{g}_1)$ must satisfy, we have $g_0=\bar{g}_0$ and $g_1=\bar{g}_1$.
\end{proof}
Note that $\bar{g}_0$ and $\bar{g}_1$ are denoted as least favorable densities \it only \rm in the sense that they are solutions to the equations in \eqref{equation62t}. In the sequel, the statistical test based on the likelihood ratio $\bar{g}_1/\bar{g}_0$ will be denoted as the (a)-test. The property defined by \eqref{equation62t} will be used in the next sections.
\subsection{Other distances}
The distance $D$ can be chosen in various ways based on mathematical tractability or the practical application \cite{Gibbs}. Symmetric distances are preferable due to their nice properties; for instance, the symmetric version of the relative entropy $D(f_0,f_1)+D(f_1,f_0)$. However, this distance does not yield an analytic expression for the LFDs and the decision rule as $$\ln l(y)=W(e^{z_0\delta(y)+z_1})-W(e^{z_1\delta(y)+z_2})+z_3 \delta(y)$$ needs to be solved to obtain the decision rule $\hat{\delta}(y)$  for $\hat{l}=1$, where $z_1$, $z_2$ and $z_3$ are constants and $W$ is the Lambert $W$-function. Symmetrized $\chi^2$ distance, i.e. $\chi^2(f_0,f_1)+\chi^2(f_1,f_0)$, is another example where the LFDs can be obtained analytically.
However, the relation between $y_u$ and $l_u$, and similarly between $l_l$ and $y_l$, cannot be obtained analytically. Another example for a symmetric distance is the squared Hellinger distance. This distance is more appealing as it scales in $[0,1]$ and it is mathematically tractable \cite{gul}, \cite{gul3}.\\
For various robust tests, including the relative entropy distance, the $\chi^2$ distance and the squared Hellinger distance, the likelihood ratio test is given by \eqref{equation26}. For the symmetrized $\chi^2$ distance, however, the test is slightly different as $\hat{l}/l$ is not a constant function for $\hat{\delta}=0$ and $\hat{\delta}(y)=1$, c.f. Sec.~\ref{simulations}. In general, designing a robust test is equivalent to determining $\hat{l}=\psi(f_0,f_1)$ for some suitable functional $\psi$ which accounts for the unmodeled uncertainties by the nominal model while maintaining the detection performance above a certain threshold.
\section{Robust Detection for the Composite Uncertainty Model}\label{section3}
Minimax robust tests, which are designed based on a neighborhood set, where every probability measure belonging to the set
is absolutely continuous with respect to the nominal distribution, e.g. \eqref{eq7}, \cite{gul3}, are more suitable for modeling errors
than the tests designed based on a neighborhood set, where not all distributions are absolutely continuous with respect to the nominals
e.g. \cite{hube65}; see \cite{dabak} and \cite{levy09}. In many practical applications, however, both types of uncertainties, namely
both modeling errors as well as outliers can occur and a reasonable approach is to build a single test which is uniformly minimax robust.
This can be done by combining one of Huber's clipped likelihood ratio tests \cite{hube73} with a robust test which is more suitable for modeling errors. The following proposition explains how this can be done.
\begin{prop}\label{prop31}
Let the inner uncertainty set be the extended version of \eqref{eq7}, i.e. \begin{equation}\label{eq1a}
{\mathcal{G}}_j=\{g_j:D(g_j,f_j)\leq \varepsilon_j\},\quad j\in\{0,1\}
\end{equation}
where $D$ is a convex distance (possibly different for each hypothesis), $0\leq \varepsilon_j<1$ are some numbers and $l=f_1/f_0$ is a monotone increasing function.
Assume that there exist $\hat{g}_0\in \mathcal{G}_0$ and $\hat{g}_1\in \mathcal{G}_1$ corresponding to probability measures $\hat{G}_0$ and $\hat{G}_1$, respectively, such that
\begin{align}\label{equation2}
&G_0[\hat{g}_1/\hat{g}_0\leq t]\geq \hat{G}_0[\hat{g}_1/\hat{g}_0\leq t]\,\,\, \forall t\in\mathbb{R}, \forall {g}_0\in\mathcal{G}_0\nonumber\\
&G_1[\hat{g}_1/\hat{g}_0\leq t]\leq \hat{G}_1[\hat{g}_1/\hat{g}_0\leq t]\,\,\, \forall t\in\mathbb{R}, \forall {g}_1\in\mathcal{G}_1.
\end{align}
Define the composite uncertainty sets
\begin{equation}\label{eqfg}
{\mathcal{F}}_{j}=\{Q_j|Q_j=(1-\epsilon_j)G_j+\epsilon_j H_j, H_j\in\Xi, g_j\in {\mathcal{G}}_j \},\,j\in\{0,1\},
\end{equation}
where $\Xi$ is the set of all probability measures on $(\Omega,{\mathscr{A}})$ and $0<\epsilon_0,\epsilon_1<1$. Then, there exist a pair of LFDs, $(\hat{Q}_0,\hat{Q}_1)$ which satisfy the saddle value condition
\begin{align}\label{eqasda}
&Q_0[\hat{q}_1/\hat{q}_0\leq t]\geq \hat{Q}_0[\hat{q}_1/\hat{q}_0\leq t]\,\,\, \forall t\in\mathbb{R}, \forall Q_0\in\mathcal{F}_0\nonumber\\
&Q_1[\hat{q}_1/\hat{q}_0\leq t]\leq \hat{Q}_1[\hat{q}_1/\hat{q}_0\leq t]\,\,\, \forall t\in\mathbb{R}, \forall Q_0\in\mathcal{F}_0,
\end{align}
if $\varepsilon_j$ and $\epsilon_j$ are small enough, i.e., $\mathcal{F}_0$ and $\mathcal{F}_1$ do not overlap, where $\hat{q}_0$ and $\hat{q}_1$ are the least favorable densities
\begin{align}
\hat{q}_0(y)&=(1-\epsilon_0)\hat{g}_0(y)\quad \quad \quad \quad  \mbox{for}\quad  \hat{g}_1(y)/\hat{g}_0(y)<c_u\nonumber\\
&=(1/c_u)(1-\epsilon_0)\hat{g}_1(y)\quad\,\,\mbox{for}\quad \hat{g}_1(y)/\hat{g}_0(y)\geq c_u\nonumber\\
\hat{q}_1(y)&=(1-\epsilon_1)\hat{g}_1(y)\quad \quad \quad \quad\mbox{for} \quad \hat{g}_1(y)/\hat{g}_0(y)>c_l\nonumber\\
&=c_l(1-\epsilon_1)\hat{g}_0(y)\quad \quad \quad \,\,\mbox{for} \quad \hat{g}_1(y)/\hat{g}_0(y)\leq c_l
\end{align}
corresponding to $\hat{Q}_0$ and $\hat{Q}_1$, respectively.
\end{prop}
\begin{proof}
The proof follows directly from the definition of the uncertainty sets
\begin{align}
\mathcal{F}_0&=\{Q_0:Q_0[Y<y]\geq (1-\epsilon_0)G_0[Y<y]-\nu_0\}\nonumber\\
\mathcal{F}_1&=\{Q_1:1-Q_1[Y<y]\geq (1-\epsilon_1)(1-G_1[Y<y])-\nu_1\}
\end{align}
with $\nu_0=\nu_1=0$ for $\epsilon-$contamination neighborhood \cite{hube68} and the stochastic ordering defined by Corollary.~\ref{korr1}. Only the first inequality in \eqref{eqasda} is proven as the second inequality can be proven using the same line of arguments. Let $b=(1-\epsilon_1)/(1-\epsilon_0)$. Then, for every $t>bc_u$ and $Q_0\in\mathcal{F}_0$, the event $E=[\hat{q}_1/\hat{q}_0\leq t]$ has full probability and for every $t\leq bc_l$ and $Q_0\in\mathcal{F}_0$, the event $E$ has null probability. Hence, \eqref{eqasda} is trivially true for these cases. For $bc_l<t\leq bc_u$, assume that the likelihood ratio $\hat{g}_1/\hat{g}_0$ is non-decreasing, which is true when $l$ is monotone and the distance is either one of Huber's distances \cite[p. 271]{hube81} or any distance with the likelihood ratio given by Eq.~\ref{equation26}, or in general a distance which results in a non-decreasing $\hat{l}=\hat{g}_1/\hat{g}_0$ for monotone $l$. Then, by Corollary~\ref{korr1} it follows that $G_0[Y\leq t]\geq \hat{G}_0[Y\leq t]$ for all $t=\hat{l}^{-1}(y)$. Let $\hat{Q}_0[Y\leq t]:=(1-\epsilon_0)\hat{G}_0[Y\leq t]$. Obviously $Q_0[
Y\leq t]\geq \hat{Q}_0[Y\leq t]$ for all $bc_l<t\leq bc_u$ and $Q_0\in\mathcal{F}_0$. Note that for non-decreasing $\hat{g}_1/\hat{g}_0$, $\hat{q}_1/\hat{q}_0$ is also non-decreasing. Hence, again by Corollary~\ref{korr1} we get $Q_0[\hat{q}_1/\hat{q}_0\leq t]\geq \hat{Q}_0[\hat{q}_1/\hat{q}_0\leq t]$ for all $t$ and $Q_0\in\mathcal{F}_0$ as claimed.
\end{proof}
The proof is independent of the choice of $D$ as long as the LFDs exist. When $D$ is the relative entropy, it follows that
\begin{align}
\hat{G}_0[\hat{g}_1/\hat{g}_0\leq \rho]&=\int_{[\hat{g}_1/\hat{g}_0<\rho]}\hat{g}_0\mbox{d}\mu+\int_{[\hat{g}_1/\hat{g}_0=\rho]}\hat{\delta}\hat{g}_0\mbox{d}\mu\nonumber\\
&\leq\int_{[\hat{g}_1/\hat{g}_0<\rho]}g_0\mbox{d}\mu+\int_{[\hat{g}_1/\hat{g}_0=\rho]}\hat{\delta}g_0\mbox{d}\mu\nonumber\\
&=G_0[\hat{g}_1/\hat{g}_0\leq \rho]
\end{align}
and in a similar way $G_1[\hat{g}_1/\hat{g}_0\leq \rho]\leq\hat{G}_1[\hat{g}_1/\hat{g}_0\leq \rho]$. This proves that the uncertainty sets based on the $\epsilon$-contamination model and the relative entropy can be combined into a composite uncertainty set
\eqref{eqfg} which accepts LFDs, $\hat{Q}_0$ and $\hat{Q}_1$ satisfying \eqref{eqasda}. Clearly, the same conclusions hold when $\nu_0$ and $\nu_1$ are non-zero.
This includes the total variation distance as a special case with $\epsilon_0=\epsilon_1=0$. Note that Prop.~\ref{prop31} is general for all thresholds.
However, when the inner uncertainty set is the KL-divergence, the decision rule $\hat{\delta}$ must be used to guarantee minimax robustness.
For a comparison, one can see that the composite model proposed in \cite{gul4} is robust only against outliers, with some flexibility, while the composite model proposed in this work is robust
against both modeling errors as well as outliers.
The LFDs, corresponding to the composite model based on the relative entropy distance, can also be obtained as
%Notice that the monotonicity of $\hat{l}$ is not required in this case. By \cite{hube65}, later stated by \cite{hube68} and \cite{hube73},
%it is also known that for any fixed $(G_0,G_1)$, \eqref{eqasda} is
\begin{align}\label{equationlfd}
\hat{q}_0(y) &= \begin{cases}\frac{(1-\epsilon_0) l_l}{z(l_l,l_u)}f_0(y), &l(y)<l_l  \\
\frac{(1-\epsilon_0)}{z(l_l,l_u)}\left(l_l^{-1}l(y)\right)^{\frac{\ln\left(k(l_l,l_u)\right)}{\ln\left(l_u/l_l\right)}}f_1(y), & l_l \leq l(y)\leq l_u \\
\frac{(1-\epsilon_0)l_u k(l_l,l_u)}{z(l_l,l_u)}f_0(y), & l_u<l(y)<c_u\\
\frac{(1-\epsilon_0)k(l_l,l_u)}{c_u z(l_l,l_u)}f_1(y), & l(y)\geq c_u\\
 \end{cases}\nonumber\\
\hat{q}_1(y) &= \begin{cases}\frac{c_l(1-\epsilon_1) l_l}{z(l_l,l_u)}f_0(y), & l(y)\leq c_l\\
 \frac{(1-\epsilon_1)}{z(l_l,l_u)}f_1(y), &c_l<l(y)<l_l\\
\frac{(1-\epsilon_1)}{z(l_l,l_u)}\left(l_l^{-1}l(y)\right)^{\frac{\ln\left(k(l_l,l_u)\right)}{\ln\left(l_u/l_l\right)}}f_1(y), & l_l \leq l(y)\leq l_u \\
\frac{(1-\epsilon_1)k(l_l,l_u)}{z(l_l,l_u)}f_1(y), & l(y)>l_u \end{cases}
\end{align}
with the corresponding likelihood ratio
\begin{align}
\hat{l}(y) &= \begin{cases}b c_l, &l(y)\leq c_l  \\
\frac{b l(y)}{l_l}, & c_l < l(y)< l_l \\
b, & l_l \leq l(y)\leq l_u\\
\frac{b l(y)}{l_u}, & l_u<l(y)<c_u\\
b c_u, & l(y)\geq c_u\\
\end{cases}
\end{align}
The choice of $D$ can be adjusted depending on the application.
For instance symmetrized $\chi^2$ distance can be preferred if the tail structure is expected to be roughly preserved. It is also not difficult to see that for variety of distances, \eqref{equationlfd} remains the same. However, special care should be taken for the choice of $b$, since it is equivalent to $\rho$. In the sequel, $D$ will be assumed to be KL-divergence with the LFDs given by \eqref{equation24} unless mentioned otherwise. For example the parameters $\epsilon_0=\epsilon_1=0$ indicate a pure KL-divergence uncertainty set with the corresponding LFDs denoted by $\hat{Q}_0:=\hat{G}_0$ and $\hat{Q}_1:=\hat{G}_1$. In the following, the corresponding test $\mathrm{d}\hat{G}_1/\mathrm{d}\hat{G}_0$ will be denoted as the (m)-test and similarly, the minimax robust test for $\varepsilon_0=\varepsilon_1=0$ will be denoted as the (h)-test and the composite test will be denoted as the (c)-test.
\subsection{Distribution of the log-likelihood ratios of LFDs}
In order to gain further insights about the minimax robust tests and to evaluate their performance, it is desirable to have the density function of the log likelihood ratio of the LFDs, i.e. $h_j^*\sim\ln \hat{q}_1/\hat{q}_0(Y)$, when $Y\sim\hat{Q}_j$, as a function of the density function of the log likelihood ratio of the nominal distributions $h_j\sim\ln f_1/f_0(Y)$ when $Y\sim F_j$, $j\in\{0,1\}$. Then, for the (h)-test, it follows that
\begin{align}\label{eol}
h_i^*(x)=&r_i^0 {\delta_x}(x-\ln (bc_l))\nonumber\\
+&(1-\epsilon_i)h_i(x-\ln b){\bf 1}_{\{\ln (bc_l)<x<\ln (bc_u)\}}\nonumber\\
+&r_i^1 {\delta_x}(x-\ln (bc_u)),\quad i\in\{0,1\}
\end{align}
where $\delta_x$ is a dirac delta function and
\begin{align}\label{eol1t}
r_0^0=(1-\epsilon_0)F_0[f_1/f_0\leq c_l],\,\,r_0^1=\frac{1}{c_u}(1-\epsilon_0)F_1[f_1/f_0\geq c_u]\nonumber\\
r_1^0=c_l(1-\epsilon_1)F_0[f_1/f_0\leq c_l],\,\,r_1^1=(1-\epsilon_1)F_1[f_1/f_0\geq c_u]
\end{align}
Similarly, for the (m)-test,
\begin{align}\label{eol2}
h_0^*(x)=&\frac{l_l}{z(l_l,l_u)}h_0(x+\ln l_l){\bf 1}_{\{x<0\}}+\frac{r}{z(l_l,l_u)}{\delta_x}(x)\nonumber\\
+&\frac{l_u k(l_l,l_u)}{z(l_l,l_u)}h_0(x+\ln l_u){\bf 1}_{\{x>0\}}\nonumber\\
h_1^*(x)=&\frac{1}{z(l_l,l_u)}h_1(x+\ln l_l){\bf 1}_{\{x<0\}}+\frac{r}{z(l_l,l_u)}{\delta_x}(x)\nonumber\\
+&\frac{k(l_l,l_u)}{z(l_l,l_u)}h_1(x+\ln l_u){\bf 1}_{\{x>0\}}
\end{align}
where
\begin{equation}
r=\int_{l_l<l<l_u}\left(l_l^{-1}l\right)^{\frac{\ln\left(k(l_l,l_u)\right)}{\ln\left(l_u/l_l\right)}}f_1\mbox{d}\mu.
\end{equation}
It can be seen that Huber's test ((h)-test) creates two point masses at the clipping thresholds $(\ln(bc_l),\ln(b_cu))$ and between them the density
of the log-likelihood ratio of the nominal distributions is shifted by $\ln b$. The robust test based on modeling errors ((m)-test),
on the other hand, shifts the density of the log-likelihood ratio of the nominal distributions $(h_0,h_1)$ by $|\ln l_l|$ to the right and adds
another part of the same density, which is shifted by $\ln l_u$, to the left. The total loss of area due to the shifting is stacked as a point mass at $x=0$.\\
The equations \eqref{eol} and \eqref{eol2} are of particular importance, first in calculating the false alarm and miss detection probabilities $\int_{t}^\infty h_0^*(x)\mbox{d}x$ and $\int_{-\infty}^t h_1^*(x)\mbox{d}x$, respectively, and second in finding the approximate
distribution of the test statistic $S_n=\sum_{i=1}^n\ln \hat{l}(Y_i)$, for $n$ independent r.v.s $Y_1,Y_2,\ldots,Y_n$, in terms of nominal distributions. However, to calculate the false alarm and miss detection probabilities, the factor of randomization, $\delta$ in Eq.~\eqref{eol2}, needs to be taken into account. That is, the contribution of the point mass at $x=0$ to the false alarm and miss detection probabilities needs to be determined.

\subsection{Limiting robustness parameters for the (m)-test}
The composite hypotheses start overlapping when the LFDs become identical. For the (m)-test, this occurs when $\mathcal{R}_1$ and $\mathcal{R}_3$ are empty sets. Let $u=1+\ln(k(l_l,l_u))/\ln (l_u/l_l)$, $w(y;u)=f_{1}(y)^{u} f_0(y)^{1-u}$ and $k(u)=\int_{\mathbb{R}}w(y;u)\mathrm{d}y$. Then, equations \eqref{equation29} and \eqref{equation30} reduce to
\begin{equation}\label{eq32}
\hspace{-1.5mm}\varepsilon_j(u)={-\ln k(u)}+\frac{u-j}{k(u)}\int_{\mathbb{R}}w(y;u)\ln l(y)\mathrm{d}y,\,\, j\in\{0,1\}
\end{equation}
\begin{prop}\label{prop4}
$\varepsilon_0$ is monotone increasing in $u$ and $\varepsilon_1$ is monotone decreasing in $u$. Hence, $0\leq \varepsilon_0 \leq D(f_1,f_0)$ and $0\leq \varepsilon_1\leq D(f_0,f_1)$.
\end{prop}
\begin{proof}
For $j=0$, it follows that
\begin{equation}\label{eq33}
\varepsilon(u)={-\ln\left(k(u)\right)}+\frac{u}{k(u)}\int_{\mathbb{R}}l(y)^{u}\ln(l(y))f_0(y)\mathrm{d}y\nonumber\\
\end{equation}
After manipulation, the first derivative of $\varepsilon(u)$ is
\begin{align}\label{eq34}
\frac{\partial{\varepsilon(u)}}{\partial u}=&\frac{u}{k(u)^2}\Bigg[k(u)\int_{\mathbb{R}}f_0(y)l(y)^u\ln(l(y))^2\mathrm{d}y\nonumber\\
-&\frac{\partial k(u)}{\partial u}\int_{\mathbb{R}}l(y)^u f_0(y)\ln(l(y))\mathrm{d}y\Bigg].
\end{align}
Inserting $k(u)$ and $\partial k(u)/\partial u$ and rearranging the terms yields
\begin{align}\label{eq35}
\frac{k(u)^2\partial{\varepsilon(u)}}{u\partial u}&=\int_{\mathbb{R}} l(y)^u f_0(y) \mathrm{d}y\int_{\mathbb{R}} l(y)^u f_0(y)\ln(l(y))^2 \mathrm{d}y\nonumber\\
-&\int_{\mathbb{R}} l(y)^u f_0(y)\ln(l(y)) \mathrm{d}y\int_{\mathbb{R}} l(y)^u f_0(y)\ln(l(y)) \mathrm{d}y\nonumber\\
=&\int_{\mathbb{R}} w(y;u) \mathrm{d}y\int_{\mathbb{R}} w(y;u)\ln(l(y))^2 \mathrm{d}y\nonumber\\
-&\left(\int_{\mathbb{R}} w(y;u)\ln(l(y)) \mathrm{d}y\right)^2
\end{align}
By H\"{o}lder's inequality, $w(y;u)$ is integrable over $\mathbb{R}$. Consider the weighted $L^2$ space, $L^2_w(\mathbb{R})$ equipped with the inner product
\begin{equation}\label{eq36}
(g,h)_w:=\frac{\int_{\mathbb{R}}g(y)h(y)w(y)\mathrm{d}y}{\int_{\mathbb{R}}w(y)\mathrm{d}y}
\end{equation}
and the resulting norm $||g||_w=\sqrt{(g,g)_w}$. By definition, $g$ is in $L_w^2$ if $g^2w$ is integrable over $\mathbb{R}$. Let $g(y)=\ln(l(y))$. Dividing \eqref{eq35} by $(\int_{\mathbb{R}}w(y)\mathrm{d}y)^2$ reads
\begin{align}\label{eq37}
\frac{k(u)^2\partial{\varepsilon(u)}}{u(\int_{\mathbb{R}}w(y)\mathrm{d}y)^2\partial u}=&||g||_w^2-(g,1)_w^2\nonumber\\
=&||g||_w^2||1||_w^2-(g,1)_w^2>0
\end{align}
The inequality follows from the Cauchy-Schwarz inequality for the inner product space $(g,h)_w$ and it is strict since $g$ and $1$ are linearly independent. What remains to be shown is that $g$ belongs to $L_w^2$, i.e., $\int_{\mathbb{R}}g(y)^2w(y)\mathrm{d}y<\infty$. If $g$ is bounded, the claim is obvious. If not, then, either $\lim_{|y|\rightarrow\infty}l(y)=\infty$ or $\lim_{|y|\rightarrow-\infty}l(y)=0$. Assume $\lim_{y\rightarrow\infty}l(y)=\infty$ and write
\begin{align}\label{eq38}
\ln(l(y))^2w(y)&=(\ln(l(y)))^2 l(y)^u f_0(y)\nonumber\\
&=\frac{(\ln l(y))^2}{l(y)^{\frac{1-u}{2}}}l(y)^{\frac{1+u}{2}} f_0(y)\nonumber\\
&=\frac{(\ln l(y))^2}{l(y)^{\frac{1-u}{2}}}f_1(y)^{\frac{1+u}{2}}f_0(y)^{\frac{1-u}{2}}
\end{align}
By H\"{o}lder the function $f_1(y)^{\frac{1+u}{2}}f_0(y)^{\frac{1-u}{2}}$ is integrable and since
\begin{equation}\label{eq39}
\lim_{|y|\rightarrow \infty}\frac{(\ln l(y))^2}{l(y)^{\frac{1-u}{2}}}=0,
\end{equation}
$g(y)^2 w(y)$ is integrable over $[0,\infty)$ by comparison with $f_1(y)^{\frac{1+u}{2}}f_0(y)^{\frac{1-u}{2}}$. If $\lim_{y\rightarrow -\infty} l(y)>0$, then $g$ is bounded on $(-\infty,0]$ and integrability over $(-\infty,0]$ follows. If $\lim_{y\rightarrow \infty} l(y)=0$, then as
\begin{equation}\label{eq40}
\ln(l(y))^2w(y)=(\ln(l(y)))^2 l(y)^u f_0(y),
\end{equation}
we have
\begin{equation}\label{eq41}
\lim_{y\rightarrow -\infty}=(\ln(l(y)))^2 l(y)^u=0,
\end{equation}
and integrability over $(-\infty,0]$ follows by comparison with $f_0$. In a similar way, $g(y)^2 w(y)$ is integrable over $\mathbb{R}$ if $\lim_{y\rightarrow-\infty}l(y)=\infty$ or $\lim_{y\rightarrow\infty}l(y)=0$. This completes the proof that $\partial{\varepsilon(u)}/\partial u>0$ and hence $\varepsilon_0\leq \varepsilon_0(1)=D(f_1,f_0)$. For $\varepsilon_1$, let $u^{'}=1-u$, $f_1:=f_0$ and $f_0:=f_1$. This gives $\varepsilon_1(u^{'})=\varepsilon_0(u^{'})$, which implies that $\varepsilon_1(u^{'})$ is increasing, therefore $\varepsilon_1(u)$ is decreasing. Note that for $\varepsilon_1$, with the substitutions of the densities, $l$ becomes decreasing, however $g$ still belongs to $L_w^2$ and the proof is complete.
\end{proof}
Prop.~\eqref{prop4} implies that \eqref{eq32} has a unique solution for all $\varepsilon_j\leq D(f_{1-j},f_j)$, $j\in\{0,1\}$. In particular, given a certain choice of $\varepsilon_j$, the solution of Eq.~\eqref{eq32} leads to $0\leq u^*\leq 1$. The corresponding maximum $\varepsilon_{1-j}$ is therefore obtained by $\varepsilon_{1-j}(u^*)$. From \eqref{eq32}, it also follows that
\begin{equation}\label{eq42}
\varepsilon_0(u)-\varepsilon_1(u)=\frac{1}{k(u)}\int_{\mathbb{R}}l(y)^u f_0(y)\ln(l(y))\mbox{d}y,
\end{equation}
which is bounded as $-D(f_0,f_1)\leq \varepsilon_0(u)-\varepsilon_1(u) \leq D(f_1,f_0)$ due to monotonicity. When $\varepsilon=\varepsilon_0(u)=\varepsilon_1(u)$, this reduces to
\begin{equation}\label{eq42a}
\varepsilon=\sup_{0 \leq u\leq 1}-\ln \int_{\mathbb{R}} f_1(y)^u f_0(y)^{1-u}\mbox{d}y
\end{equation}
which is the \it Chernoff distance \rm and if additionally $f_0(y)=f_1(-y)\forall y$, it further reduces to
\begin{equation}\label{eq42b}
\varepsilon=-\ln\int_{\mathbb{R}}\sqrt{f_0(y)f_1(y)}\mathrm{d}y,
\end{equation}
which is the \it Bhattacharyya distance \rm between the nominal densities.
\subsection{Limiting robustness parameters for the (h)-test}
\begin{prop}\label{proppp}
The maximum achievable pair of $(\epsilon_0,\epsilon_1)$ with respect to the $\epsilon$-contamination model are obtained by
\begin{align}
(1-\epsilon_0)(P_0[l\leq 1/b]-bP_1[l\leq 1/b])=\epsilon_1\label{equation51}
\end{align}
where $b=(1-\epsilon_1)/(1-\epsilon_0)$.
\end{prop}
\begin{proof}
By Huber \cite{hube65}, it is known that $h_1(c_l)=P_1[p_1/p_0>c_l]+c_lP_0[p_1/p_0\leq c_l]$ is an increasing function of $c_l$ and $h_2(c_u)=P_0[p_1/p_0<c_u]+(1/c_u)P_1[p_1/p_0\geq c_u]$, in a similar manner, is a decreasing function of $c_u$. This implies that $\epsilon_0$ and $\epsilon_1$ are maximized when $c_l$ is maximized and $c_u$ is minimized. The maximum of $c_l$ is equal to the minimum of $c_u$ such that the hypotheses do not overlap. As a result for $c=c_l=c_u$, it follows that $\hat{l}(y)=bc$ for all $y\in\mathbb{R}$. Since no density is greater than any other for all $y\in\mathbb{R}$, the conclusion is that $c=1/b$. Rewriting the equations, $h_1(c:=1/b)=1/(1-\epsilon_1)$ or equivalently $h_2(c:=1/b)=1/(1-\epsilon_0)$, completes the proof.
\end{proof} Let $u_0=\mbox{ess}\inf_{[\mu]} l$, $u_1=\mbox{ess}\sup_{[\mu]} l$, and let $k=1-\epsilon_0$ be known and $u=1/(1-\epsilon_1)$ to be determined. With these substitutions \eqref{equation51} can be written as $f(u)=kuP_0[l\leq ku]-P_1[l\leq ku]-u+1$.
\begin{lem}\label{lemma34}
The function $f$ is continuous, $f(u)=1-u$ for $0\leq ku\leq u_0 $, is strictly decreasing for $ku<u_1$, tends to $-\infty$
for $k<1$ and tends to $0$ for $k=1$ as $u\uparrow\infty$.
\end{lem}
\begin{proof}
we have
\begin{align}
f(u)=&kuP_0[l\leq ku]-P_1[l\leq ku]-u+1\nonumber\\
=&\int_{\{l\leq ku\}}(ku-l)p_0\mbox{d}\mu-u+1.
\end{align}
Thus,
\begin{align}
f(u+\Delta)-f(u)&=\int_{\{ku< l\leq k(u+\Delta)\}}\left(ku+k\Delta-l\right)p_0\mbox{d}\mu\nonumber\\
&+\Delta\left(k\int_{\{l\leq ku\}}p_0\mbox{d}\mu-1\right).
\end{align}
It then follows that
\begin{align}
f(u+\Delta)-f(u)&\leq k\Delta P_0[ku<l\leq k(u+\Delta)]\nonumber\\
&+k\Delta P_0[l\leq ku]-\Delta\nonumber\\
&=\Delta\left(k P_0[l\leq k(u+\Delta)]-1\right)\leq 0
\end{align}
for any positive $\Delta$ and for all $u\geq 0$. Since $f(u+\Delta)-f(u)\geq \Delta\left(k P_0[l<ku]-1\right)\geq -\Delta$, the conclusion is that
$0\geq f(u+\Delta)-f(u)\geq -\Delta$ from where continuity and monotonicity follow. For $ku\leq k(u+\Delta)<u_1$, it also follows that
$f(u+\Delta)-f(u)<0$, hence, $f(u)$ is strictly decreasing, tends to $-\infty$ for $k<1$ and tends to $0$ for $k=1$ as $u\uparrow \infty$.\\
\end{proof}
Lemma~\ref{lemma34} implies that $\epsilon_1$ can be determined uniquely for all $0 \leq\epsilon_0< 1$.
Moreover, Lemma~\ref{lemma34} extends to the case when $\epsilon_1$ is known and $\epsilon_0$ is variable due to the duality of the parameters, $\epsilon_0$ and $\epsilon_1$.

\section{Fixed sample size tests}\label{section4}
The robust version of the likelihood ratio test with respect to the uncertainty model \eqref{eq7} can be generalized to $n$ independent samples, i.e.
\begin{equation}\label{eq43}
\hat{l}(\mathbf{y})=\prod_{i=1}^n \hat{l}(y_i)\stackrel{{\cal{H}}_1}{\underset{{\cal{H}}_0}{\gtrless}}\gamma,
\end{equation}
which is equivalent to the nominal likelihood ratio test
\begin{equation}\label{eq44}
\prod_{i=1}^n l(y_i)\stackrel{{\cal{H}}_1}{\underset{{\cal{H}}_0}{\gtrless}}\left(\frac{l_u}{l_l}\right)^{\sum_{i=1}^n \hat{\delta}_i(y)}(l_l)^{n}\gamma
\end{equation}
and similarly in the logarithmic scale
\begin{equation}\label{equation58}
\sum_{i=1}^ n \ln \left(l(y_i)(l_l/l_u)^{\hat{\delta}_i(y)}\right)\stackrel{{\cal{H}}_1}{\underset{{\cal{H}}_0}{\gtrless}}n\ln l_l,
\end{equation}
for $\gamma=1$. Given the upper and lower thresholds, $l_l$ and $l_u$, if $\sum_{i=1}^n \hat{\delta}(y_i)\approx 0$, the original threshold of the nominal test is moved from $0$ to $n\ln l_l$, increasing the false alarm probability. Similarly, if $\sum_{i=1}^n \hat{\delta}(y_i)\approx 1$, the original threshold of the nominal test is moved to $n\ln l_u$, which increases the miss detection probability. Let $\mathbf{y}=[y_1,\ldots,y_n]$ be the observation vector and $b=1$. Assume that there are $n_1$ (and $n_2$) observations in $\mathbf{y}$ whose likelihood ratios are clipped to ($c_l$) (and $c_u$), respectively. Then, Huber's clipped likelihood ratio test can be represented in the log domain as
\begin{equation}\label{equation59}
\sum_{i=1}^{n-n_1-n_2} \ln \left(l(y_i)\right)\stackrel{{\cal{H}}_1}{\underset{{\cal{H}}_0}{\gtrless}}-(n_1\ln c_u+n_2\ln c_l).
\end{equation}
Eventually, the robust test based on the composite model \eqref{eqfg} can be given by
\begin{align}\label{equation60}
\sum_{i=1}^{n-n_1-n_2} \ln \left(l(y_i){l_l}^{1-\delta_i(y)}{l_u}^{-\hat{\delta}_i(y)}\right)\stackrel{{\cal{H}}_1}{\underset{{\cal{H}}_0}{\gtrless}}-(n_1\ln c_u+n_2\ln c_l).
\end{align}
where $n_1$ and $n_2$ are now due to clipping of the likelihood ratio given by \eqref{equation26}. The composite test combines the robustness properties of both the clipped likelihood ratio test as well as the robust test for modeling errors. Single sample robust tests are extended to multiple samples through multiplication of the likelihood ratios due to the independency of every measurable set of observations. Unlike Huber's robust test, there is no stochastic ordering for the LFDs of modeling errors. Hence, the composite model can be expected to be robust, but minimax robustness is not guaranteed for $n>1$.
\subsection{Asymptotic performance analysis}
Large deviations theory can be used to analyze the asymptotic performance of the robust tests. Consider the following theorem by Cram\'{e}r \cite{Cramer}:
\begin{thm}[Cram\'{e}r]\label{cramer}
Let $(Y_i)_{i\geq 1}$ be a sequence of i.i.d. random variables, $S_n=\frac{1}{n}\sum_{i=1}^n Y_i$ be their average sum and $M_{Y_1}(u):=E[e^{uY_1}]<\infty$ be the moment generating function of the r.v. $Y_1$. Then, for all $t>E[Y_1]$
\begin{equation}
\lim_{n\rightarrow \infty}\frac{1}{n}\ln P(S_n\geq t)=-I(t)
\end{equation}
where the rate function $I$ is defined by
\begin{equation}\label{equation65}
I(t):=\sup_{u}\left(tu-\ln M_{Y_1}(u)\right),
\end{equation}
which is the Legendre transform of the log moment generating function.
\end{thm}
\begin{rem}\label{rem2}
Theorem \ref{cramer} implies
\begin{equation}
\lim_{n\rightarrow \infty}\frac{1}{n}\ln P(S_n< t)=-I(t)
\end{equation}
for all $t<E[Y_1]$. To see this, take $X_i=-Y_i$ and consider
\begin{equation}
P\left(\frac{1}{n}\sum_{i=1}^n X_i>-t\right).
\end{equation}
Applying Cram\'{e}r's theorem to the r.v. $X_1$ and the threshold $-t$, it follows that $M_{Y_1}(u)=M_{X_1}(-u)$ and
\begin{equation}\label{equation655}
I(t):=\sup_{u}\left(tu-\ln M_{Y_1}(u)\right)=\sup_{u}\left(-tu-\ln M_{X_1}(-u)\right)
\end{equation}
\end{rem}
Let $S_n^*:=\frac{1}{n}\sum_{i=1}^n \ln \hat{l}(Y_i)$ with $\hat{l}(Y_i)=\hat{q}_1(Y_i)/\hat{q}_0(Y_i)$ for $Y_i\sim Q_0$ under $\mathcal{H}_0$ and for $Y_i\sim Q_1$ under $\mathcal{H}_1$ for all $i\in\{0,\ldots,n\}$. Furthermore, let the first and second type of error probabilities defined to be $P_E^0(t)=P(S_n^*\geq t)$ and $P_E^1(t)=P(S_n^*<t)$. Then, for all $E_{Q_0}[\hat{l}(Y_1)]<t<E_{Q_1}[\hat{l}(Y_1)]$ from Theorem~\ref{cramer} and Remark~\ref{rem2},
\begin{equation}\label{equation656}
\lim_{n\rightarrow \infty}\frac{1}{n}\ln P_E^j(t)=-I_j(t)\quad j=0,1,
\end{equation}
where
\begin{equation}\label{equation657}
I_j(t):=\sup_{u}\left(tu-\ln M_{Y_1}^j(u)\right)\quad j=0,1,
\end{equation}
with
\begin{equation}\label{equation658}
M_{Y_1}^j(u)=\int_{\mathbb{R}}\hat{l}(y)q_j(y)\mathrm{d}y\quad j=0,1.
\end{equation}
\begin{rem}\label{rem3}
Interestingly, if $\hat{q}_1=q_1$ and $\hat{q}_0=q_0$, the parametric curve $(\varepsilon_0(u),\varepsilon_1(u))$ for $0 \leq u\leq 1$, \eqref{eq32} with $f_0:=q_0$ and $f_1:=q_1$ implies $(I_0(t),I_1(t))$ for all $E_{Q_0}[\hat{l}(Y_1)]<t<E_{Q_1}[\hat{l}(Y_1)]$. To prove this claim, observe that in this case we have $M_{Y_1}^1(u)=M_{Y_1}^0(u+1)$. Applying this result to \eqref{equation657}, taking the derivative of $tu-\ln M_{Y_1}^j(u)$ with respect to $u$, and rewriting $I_j$ in terms of maximizing $u$ gives \eqref{eq32} with the aforementioned substitutions. Since the mapping from $0 \leq u\leq 1$ to $E_{Q_0}[\hat{l}(Y_1)]<t<E_{Q_1}[\hat{l}(Y_i)]$ is bijective, as the derivative of a convex function $\ln M_{Y_1}^j(u)$ \cite[p.77]{levy} is increasing, the proof is complete.
\end{rem}

\subsection{Limiting tests}
\subsubsection{Limiting (m)-test}
The limiting case, $\lim_{l_l\rightarrow \inf l}$ and $\lim_{l_u\rightarrow \sup l}$, is of particular interest.
For a single sample, the test becomes a pure randomized test having a success probability $\delta$ which increases with $l$ \eqref{equation25}. For $n$ independent samples, assume $l_l:=1/l_u$ and consider the normalization $\ln\hat{l}^{'}(y)=(\ln l_u-\ln \hat{l}(y))/(\ln l_u-\ln l_l)$. Then, as $l_l\downarrow 0$ and $l_u\uparrow\infty$, the test statistic $\ln\hat{l}^{'}_n(y)=\sum_{i=1}^n \ln\hat{l}^{'}(y_i)$ tends to $\sum_{i=1}^n \hat{\delta}(y_i)$, which is the \it soft \rm version of the sign test.
\subsubsection{Limiting (h)-test}
The limiting test for Huber's clipped likelihood ratio test is known to be the sign test \cite{hube65}.
\subsubsection{Limiting (a)-test}
The limiting asymptotically robust test is again a likelihood ratio test with the threshold determined by $u+v\rightarrow 1$ in \eqref{equation62q}.
\section{Robust sequential probability ratio test}\label{section5}
Sequential probability ratio tests (SPRT)s can be preferable over fixed sample size tests due to their strong optimality properties \cite{levy}. Let $S_n=\sum_{i=1}^n Y_i$. Then, for given target error probabilities of the first and second kind, $\alpha$ and $\beta$ respectively, by Wald \cite{wald}, there exist an upper threshold $t_u>1$ and a lower threshold $0<t_l<1$ such that SPRT continues taking another sample if $t_l<S_n<t_u$,
terminates and decides for $\mathcal{H}_0$ if $S_n<t_l$ and decides for the alternative hypothesis $\mathcal{H}_1$ if $S_n>t_u$, for the first
time $N=\min\{n:S_n>t_u\,\, \mbox{or}\,\,S_n<t_l\}$. Furthermore, let the binary r.v. $\upsilon$ denote the decision of the sequential test, i.e. $\upsilon=1$
to decide for $\mathcal{H}_1$ and $\upsilon=0$ for $\mathcal{H}_0$. Similar to the fixed sample size test, a robust version of the sequential test can be defined in terms of the nominal likelihood ratios and modified thresholds
\begin{equation}
m(y_i)\ln t_l<\sum_{i=1}^n \ln l(y_i)<m(y_i)\ln t_u
\end{equation}
with
\begin{equation}
m(y_i)=\ln l_l\left(n-\sum_{i=1}^n \hat{\delta}(y_i)\right)+\ln l_u\sum_{i=1}^n \hat{\delta}(y_i),
\end{equation}
for the (m)-test. Extensions to the (h)-test as well as to the (c)-test for the function $m$ follow in a straightforward manner from \eqref{equation59} and \eqref{equation60}.
However, it can be observed that all three robust tests are still some subsets of a possible design which considers two possibly different functions
$m_0:\Omega\mapsto\mathbb{R}$ and $m_1:\Omega\mapsto\mathbb{R}$ as multiplicands to the lower and upper thresholds. Hence, it can be concluded that a general design of a robust sequential test is a design of two (random) functions $m_0$ and $m_1$ such that both the expected number of samples, $E[N]$, as well as the error probabilities of the first and second kind $(\alpha,\beta)$ are bounded from above for all probability measures in the vicinity of the nominal distributions defined by a neighborhood of uncertainty. In the following, the robust tests
that have already been designed or introduced are analyzed for the sequential test.
Throughout design or analysis of a robust sequential test can be found for example in \cite{sequential}, where the probability distributions are assumed to be discrete with finite set of values, or in \cite{sequential2}, where Huber's test is rigorously shown to be asymptotically robust.\\
Let $S_n=\sum_{i=1}^n \ln \hat{l}(Y_i)$ be the test statistic where $Y_1,Y_2,\ldots,Y_n$ are again i.i.d. and $\ln \hat{l}(Y_1)$ follows a probability distribution $Q_j$, which accepts a continuous density function $q_j$, when the true hypothesis is $\mathcal{H}_j$, $j\in\{0,1\}$. Let furthermore
\begin{equation}
h_{j,n}(y)=\frac{\partial}{\partial y}P_j[S_n\leq y,S_{1},\ldots,S_{n-1}\in (\ln t_l,\ln t_u)]
\end{equation}
be the density function of $S_n$ under $\mathcal{H}_j$, $j\in\{0,1\}$ when all $S_k$, $k<n$ are in $(\ln t_l,\ln t_u)$. Hence, the distribution of $N$ can be calculated recursively by
\begin{align}
P_j[N=n]=\int_{(-\infty,\ln t_l)\cup (\ln t_u,\infty)}h_{j,n}(y)\mbox{d}y\nonumber\\
h_{j,n}(y)=\int_{\ln t_l}^{\ln t_u}h_{j,n-1}(\omega)q_j(y-\omega)\mbox{d}\omega
\end{align}
with the initial condition $h_{j,1}=q_j$, $j\in\{0,1\}$, \cite{feller}. Accordingly, it follows that
\begin{equation}\label{equation73}
\mathrm{E}_j[N]=\sum_{n=1}^\infty n P_j[N=n].
\end{equation}
Slightly modifying $P_j[N=n]$ by imposing the constraint that the test will terminate either with the rejection or acceptance $\mathcal{H}_0$,
\begin{align}\label{equation74}
P_0[N=n|\upsilon=1]=&\int_{(\ln l_u,\infty)}h_{0,n}(y)\mbox{d}y\nonumber\\
P_1[N=n|\upsilon=0]=&\int_{(-\infty,\ln l_l)}h_{1,n}(y)\mbox{d}y,
\end{align}
we get
\begin{equation}\label{equation75}
\hspace{-1mm}\alpha=\sum_{n=1}^\infty P_0[N=n|\upsilon=1],\,\,\, \beta=\sum_{n=1}^\infty P_1[N=n|\upsilon=0].
\end{equation}
Herein, $\alpha$, $\beta$, and $E_j[N]$ are all implicit functions of $(t_l,t_u)$, and $Q_j$, $j\in\{0,1\}$.
When the notations are made explicit, a minimax robust sequential test must satisfy
\begin{align}\label{equation76}
\alpha_{\hat{Q}_0}[t_l,t_u]&\geq \alpha_{Q_0}[t_l,t_u]\nonumber\\
\beta_{\hat{Q}_1}[t_l,t_u]&\geq \beta_{Q_1}[t_l,t_u]
\end{align}
and
\begin{align} \label{equation77}
\mathrm{E}_{\hat{Q}_0}[N;t_l,t_u]&\geq \mathrm{E}_{Q_0}[N;t_l,t_u]\nonumber\\
\mathrm{E}_{\hat{Q}_1}[N;t_l,t_u]&\geq \mathrm{E}_{Q_1}[N;t_l,t_u]
\end{align}
for all $(Q_0,Q_1)\in \mathcal{Q}_0\times \mathcal{Q}_1$ and for all $(t_l,t_u)$.\\
The sequential (m)-test does not satisfy \eqref{equation76} and \eqref{equation77} even asymptotically, i.e. when $t_l\rightarrow 0$ and $t_u\rightarrow \infty$, or equivalently $Q_0\rightarrow Q_1$ or $Q_1\rightarrow Q_0$. This is due to the lack of stochastic ordering between $\hat{Q}_0$ and $Q_0$, likewise between $\hat{Q}_1$ and $Q_1$.\\
Similarly, the sequential (a)-test \eqref{equation61} does not satisfy \eqref{equation76} (asymptotically) either. Again, asymptotically, the behavior of the cumulative sums are determined by their non-random drift, i.e., $S_n\sim N E_Q[\ln \hat{l}(Y)]$ or $\hat{S}_n\sim N E_{\hat{Q}}[\ln \hat{l}(Y)]$ and Wald's approximations become exact, i.e., $E[S_n]\approx \ln t_l$ under $\mathcal{H}_0$ and $E[S_n]\approx \ln t_u$ under $\mathcal{H}_1$. Combining both conditions, it follows that
\begin{align}\label{equation78}
\mathrm{E}_{Q_0}[N]\sim \frac{\ln t_l}{\mathrm{E}_{Q_0}[\ln \hat{l}(Y)]},\quad \mathrm{E}_{Q_1}[N]\sim \frac{\ln t_u}{\mathrm{E}_{Q_1}[\ln \hat{l}(Y)]}.
\end{align}
From \eqref{equation62t}, it is known that \eqref{equation61} maximizes the right hand sides of \eqref{equation78}. Therefore, the sequential (a)-test satisfies \eqref{equation77} asymptotically.\\
For the sequential (h)-test, it is known that \eqref{equation76} and \eqref{equation77} are satisfied asymptotically \cite{hube65}. Additionally in \cite{hube81}, a counterexample is given, which shows that \eqref{equation77} does not hold in general, i.e., for all $(t_l,t_u)$. In the following, it is shown that the sequential (h)-test satisfies \eqref{equation76} for all $(t_l,t_u)$.\\

\begin{thm}[Coupling]\label{prop31}
Let $(X,Y)$ be a pair of random variables on $(\Omega,{\mathscr{A}},P)$ with $X\succ_{ST}Y$. On the same probability space there exist another pair of random variables $(X^{'},Y^{'})$ such that $X^{'}=X$ in distribution, $Y^{'}=Y$ in distribution and $X^{'}\geq Y^{'}$ almost surely.
\end{thm}

\begin{proof}
Take $X^{'}=X$ and $Y^{'}=G^{-1}(F(X))$. Then, $X^{'}=X$ in distribution, $P[G^{-1}(F(X))\leq x]=P[F(X)\leq G(x)]=P[X\leq F^{-1}(G(x))]=F[F^{-1}G(x)]=G(x)=:P[Y\leq x]$, so $Y^{'}=Y$ in distribution and since $P[Y^{'}\geq X^{'}]=P[G^{-1}F(X)\geq X]=P[F(X)\geq G(X)]=1$, $X^{'}\geq Y^{'}$ almost surely.
\end{proof}

\begin{prop}\label{prop33}
Let $X_i$ and $Y_i$ be two continuous random variables on $\mathbb{R}$ having distribution functions $F$ and $G$, respectively and satisfying $G(y)\geq F(y)$ for all $y$. Furthermore, let $S^X_n=\sum_{i=1}^n X_i$, $S^Y_n=\sum_{i=1}^n Y_i$, $A>0$, and $B<0$. Denote $\tau_A=\inf\{n\geq 0:S_n\geq A\}$ and $\tau_B=\inf\{n\geq 0:S_n\leq B\}$ the hitting/stopping times of $S_n$ at the upper and lower thresholds respectively. Then,
\begin{equation}
P_{S_n^X}[\tau_A>\tau_B]\geq P_{S_n^Y}[\tau_A>\tau_B].
\end{equation}
\end{prop}

\begin{proof}
For a well defined comparison, exclude the cases $X\not\equiv 0$ and $Y\not\equiv 0$ s.t. at least $\tau_A<\infty$ or $\tau_B<\infty$ almost surely and $\tau_A>\tau_B$ is well defined.
The argument $G(y)\geq F(y)$ for all $y$ implies $X\succ_{ST}Y$ and from Prop.~\ref{prop31}, there exists $(X^{'},Y^{'})$ such that $X^{'}=X$, $Y^{'}=Y$ in distribution and $X^{'}>Y^{'}$ almost surely (a.s.)
Consider the sequence of i.i.d. random variables $(X^{'}_n,Y^{'}_n)_{n\geq 1}$ s.t. $(X_1^{'},Y_1^{'})=(X^{'},Y^{'})$ in distribution.
Then, $(X^{'}_n)_{n\geq 1}=(X_n)_{n\geq 1}$ and $(Y^{'}_n)_{n\geq 1}=(Y_n)_{n\geq 1}$ in distribution. Defining $S^{X^{'}}_n=\sum_{i=1}^n X_i^{'}$ and $S^{Y^{'}}_n=\sum_{i=1}^n Y_i^{'}$,
we also have $S^{X^{'}}_n=S^{X}_n$ and $S^{Y^{'}}_n=S^{Y}_n$ in distribution. Since $X^{'}>Y^{'}$ a.s. and accordingly $X_i^{'}>Y_i^{'}$ a.s. for all $i$, $S^{X^{'}}_n\geq S^{Y^{'}}_n$ a.s. Let $\tau_A^{X^{'}}=\inf \{n\geq 0:S_n^{X^{'}}\geq A\}$ and define $\tau_A^{Y^{'}}$, $\tau_B^{X^{'}}$ and $\tau_B^{Y^{'}}$ in the same way. Then $S_n^{Y^{'}}\geq A$ implies $S_n^{X^{'}}\geq A$ for all $n$, so $\tau_A^{X^{'}}\leq \tau_A^{Y^{'}}$ and in the same way $\tau_B^{X^{'}}\geq \tau_B^{Y^{'}}$. Hence,
\begin{align}\label{equation81}
&P_{(S_n^X)} (\tau_A < \tau_B) = P (\tau_A^{X'} < \tau_B^{X'})\geq \nonumber\\
&P (\tau_A^{Y'} < \tau_B^{Y'}) = P_{(S_n^Y)} (\tau_A < \tau_B)
\end{align}
\end{proof}
Let $X\sim \hat{Q}_0$ and $Y\sim Q_0$, likewise $X\sim Q_1$ and $Y\sim \hat{Q}_1$ with $A=\ln t_u$ and $B=\ln t_l$.
Then, it is easy to see that \eqref{equation81} is equivalent to \eqref{equation76} for any pair $(t_l,t_u)$. This result includes not only the (h)-test, but also all tests in \cite{hube68}, \cite{hube73}.\\
For the expected number of samples, the requirement is
\begin{equation}
E[\min\{\tau_A^{X^{'}},\tau_B^{X^{'}}\}]\geq E[\min\{\tau_A^{Y^{'}},\tau_B^{Y^{'}}\}].
\end{equation}
This inequality does not have to hold in general. Intuitively, however, it is expected that it holds for the majority of the cases, especially when $t_l$ is small enough and $t_u$ is large enough.
\section{Robust estimation}\label{section6}
The composite uncertainty model given in equation \eqref{eqfg} extends to robust estimation problems. Let $f_{\boldsymbol\theta}$ be a nominal probability density function corresponding to the distribution function $F$ with parameters $\boldsymbol \theta=[\theta_1,\theta_2,\ldots,\theta_N]$. In a general estimation framework, some parameters, possibly a sub-vector of $\boldsymbol \theta$ can be estimated well whereas some other parameters might not be, possibly due to a fast change of the parameters with time or due to the random nature of the parameters whose distributions are unknown. It is also possible that the known parameters might deviate slightly from the true values depending on the nature of the application or without regarding the parametric model, the shape of the distribution might be slightly different than expected, e.g. when there is lack of data but the CLT is assumed. In such cases, we have modeling errors that go unmodeled in addition to the outliers caused by some unexpected
events. Therefore, it is desirable to design robust estimators which are not only able to deal with outliers but also with modeling errors, as given by \eqref{eqfg}.\\
To account for the composite model, let $T_n(F)$ be a functional $T_n:Y^n\mapsto\mathbb{R}$ of $\mathbb{R}^n$-valued random variable $Y^n$ with i.i.d. replicas following a certain distribution $F$, i.e., for $Y^n=[Y_1,Y_2,\ldots,Y_n]$ each pair of r.v.s $(Y_i,Y_k)$ with $i\neq k$ are i.i.d., having a distribution function $F$. Then, it is desirable that $\lim_{n\rightarrow\infty}T_n(F)=\theta$ for some parameter $\theta$ when $F$ is the nominal distribution. Let $F_{T_n}$ and $Q_{T_n}$ be the distribution functions of $T_n$ when $F$ and $Q$ are the distribution functions of $Y_1$, respectively. Then, it is also expected that for every $\epsilon>0$, there exist $\delta>0$ and an $n_0>0$, such that for all $n>n_0$ and $Q\in\mathcal{F}$, $D(F_{T_n},Q_{T_n})<\epsilon$ whenever $D(F,Q)<\delta$ for some metric $D$. This is a straightforward extension of Hampel's equicontinuity theorem of robustness for the composite uncertainty model. Accordingly, the influence function can be modified as
\begin{equation}
IF(y,T)=\lim_{\epsilon\rightarrow 0}\sup_{G\in\mathcal{G}}\frac{T((1-\epsilon)G+\epsilon\delta_x)-T(F)}{\epsilon}
\end{equation}
to account for the modeling errors in addition to the outliers. Similarly, the maximum bias as being another important metric to measure the robustness of an estimator can be obtained as
\begin{align}
b(\epsilon)=&\sup_{Q,G}|T(Q)-T(G)|\nonumber\\
=&\sup_{G,H}|T((1-\epsilon)G+\epsilon H)-T(G)|.
\end{align}
\section{Simulations}\label{simulations}
In this section, simulations are performed in order to visualize and validate the theoretical findings. Observations are assumed to be real valued. The formulations are general, therefore, the observation space can be any discrete, continuous, finite or infinite set, with slight modifications for the discrete case. It can also be extended to the multidimensional case, but for large $n$, Monte-Carlo simulations may be required in order to solve the non-linear equations, c.f., \cite{gul4}.\\In the first simulation, the composite uncertainty model \eqref{eqfg} with mean and variance shifted nominal distributions, $F_0\sim\mathcal{N}(-1,1)$ and $F_1\sim\mathcal{N}(1,2)$, and the uncertainty parameters $\{\varepsilon_0=0.15,\, \varepsilon_1=0.05, \epsilon_0=0.02,\,\epsilon_1=0.02\}$ is considered. Note that for this choice of nominal distributions, neither $\partial F_1/\partial F_0$ is monotone nor they are symmetric with respect to any point on their domain or codomain. In addition
to this, $\varepsilon_0\neq\varepsilon_1$ is chosen so that the given example is general enough for the solution of Equations \eqref{equation29} and \eqref{equation30}. Regarding the $\epsilon-$contamination part of the composite model, $\epsilon_0=\epsilon_1$ is chosen to be consistent with $\rho=1$ for the uncertainty model based on relative entropy. Accordingly, in Fig.~\ref{fig1} the LFDs together
with their nominal distributions are shown, whereas in Fig.~\ref{fig2} the log-likelihood ratios of the nominal distributions, the least favorable densities $(\hat{g}_0,\hat{g}_1)$ when $\{\varepsilon_0=0.15,\, \varepsilon_1=0.05, \epsilon_0=\epsilon_1=0\}$ and the least favorable densities $(\hat{q}_0,\hat{q}_1)$ when $\{\varepsilon_0=0.15,\, \varepsilon_1=0.05, \epsilon_0=0.02,\,\epsilon_1=0.02\}$ are shown.
\begin{figure}[ttt]
  \centering
  \psfrag{x}[t][]{$\epsilon=\epsilon_0=\epsilon_1$}
  \psfrag{y}{$y_u$}
  \centerline{\includegraphics[width=8.5cm]{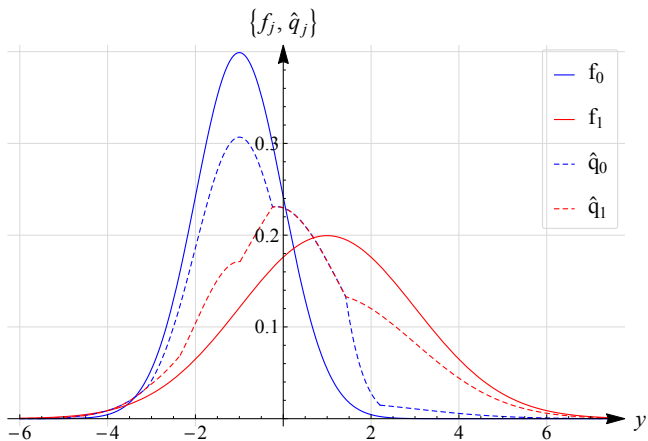}}
\vspace{-3mm}
\caption{Least favorable density functions, $\hat{q}_0$ and $\hat{q}_1$, for $\{\varepsilon_0=0.15,\, \varepsilon_1=0.05, \epsilon_0=0.02,\,\epsilon_1=0.02\}$ together with their corresponding nominal density functions, $f_0$ and $f_1$, for $F_0\sim\mathcal{N}(-1,1)$ and $F_1\sim\mathcal{N}(1,2)$.}\label{fig1}
\end{figure}
\begin{figure}[ttt]
  \centering
  \psfrag{x}[t][]{$\epsilon=\epsilon_0=\epsilon_1$}
  \psfrag{y}{$y_u$}
  \centerline{\includegraphics[width=8.5cm]{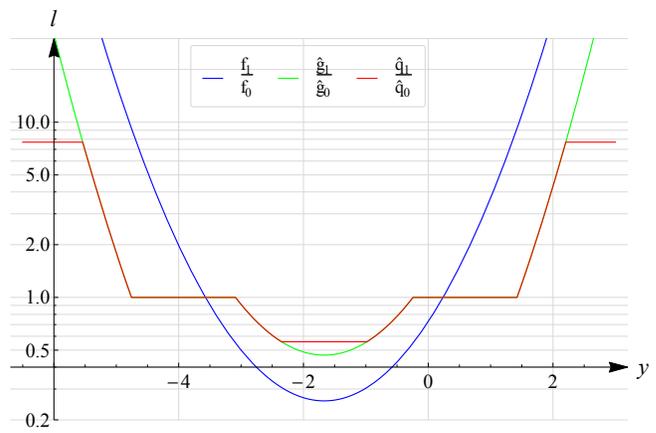}}
\vspace{-3mm}
\caption{Likelihood ratios of the nominal density functions, $f_0$ and $f_1$, least favorable density functions, $\hat{g}_0$ and $\hat{g}_1$, with $\{\varepsilon_0=0.15,\, \varepsilon_1=0.05\}$ and composite least favorable density functions, $\hat{q}_0$ and $\hat{q}_1$, with $\{\varepsilon_0=0.15,\, \varepsilon_1=0.05,\epsilon_0=0.02,\,\epsilon_1=0.02\}$ for $F_0\sim\mathcal{N}(-1,1)$ and $F_1\sim\mathcal{N}(1,2)$.}\label{fig2}
\end{figure}\\In the second simulation, the mean shifted Gaussian distributions $F_0\sim\mathcal{N}(-1,1)$ and $F_1\sim\mathcal{N}(1,1)$ are considered when the closed balls are formed with respect to the symmetrized $\chi^2$ distance with $(\varepsilon_0=\varepsilon_1=0.08)$ and a relative entropy distance $D$ with $(\varepsilon_0=\varepsilon_1\approx 0.0087)$. The parameters are chosen such, such that the LFDs resulting from both distances have equal relative entropy relative to the nominal density functions. Figure~\ref{fig3} illustrates $\hat{l}/l$, the ratio of the likelihood ratios. It can be seen that there is a significant difference when the $\chi^2$ distance is considered instead of the KL-divergence. While this ratio tends to $1$ as $\hat{\delta}\rightarrow 0$ and $\hat{\delta}\rightarrow 1$ for the symmetrized $\chi^2$ distance, meaning that the tails of the density functions are preserved, it is a constant $l_l<1$ when $\hat{\delta}=0$ and another constant $l_u>1$ when $\hat{\delta}=1$ for the KL-divergence.
\begin{figure}[ttt]
\centering
\psfrag{x}[t][]{$\epsilon=\epsilon_0=\epsilon_1$}
\psfrag{y}{$y_u$}
\centerline{\includegraphics[width=8.5cm]{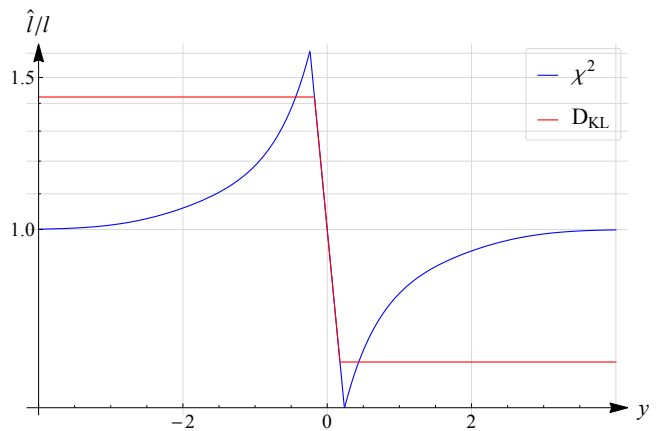}}
\vspace{-3mm}
\caption{The ratio of the likelihood ratio of the least favorable densities $\hat{l}=\hat{g}_1/\hat{g}_0$ to the likelihood ratio of the nominal distributions $l=f_1/f_0$ for $F_0\sim\mathcal{N}(-1,1)$ and $F_1\sim\mathcal{N}(1,1)$, when the LFDs are based on the symmetrized $\chi^2$ distance with $\varepsilon=\varepsilon_0=\varepsilon_1\approx 0.08$ and when the LFDs are based on the KL-divergence with $\varepsilon\approx 0.0087$.}\label{fig3}
\end{figure}\\
In the third simulation, again the same mean shifted Gaussian distributions are considered. Of interest is the curvature of the maximum robustness parameters for the (h)-test \eqref{equation51} versus the (m)-test \eqref{eq32}. Figure~\ref{fig4} illustrates the outcome of this simulation.
\begin{figure}[ttt]
  \centering
  \psfrag{x}[t][]{$\epsilon=\epsilon_0=\epsilon_1$}
  \psfrag{y}{$y_u$}
  \centerline{\includegraphics[width=8.5cm]{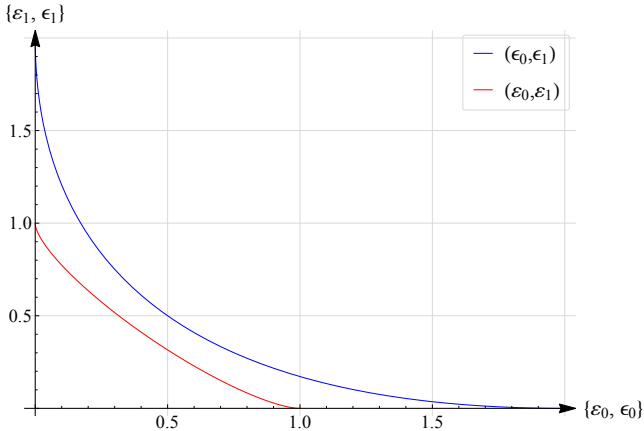}}
\vspace{-3mm}
\caption{Maximum achievable robustness parameters with respect to the (h)-test and the (m)-test when the nominal distributions are $F_0\sim\mathcal{N}(-1,1)$ and $F_1\sim\mathcal{N}(1,1)$.\hspace{-3mm}}\label{fig4}
\end{figure}\\
In the fourth simulation, asymptotic decrease rates, $I_0$ and $I_1$ \eqref{equation657}, of the type I and type II error probabilities are considered. The log-likelihood ratio test is built based on LFDs of the composite model $\hat{l}(Y)=\hat{q}_1(Y)/\hat{q}_0(Y)$  with parameters $\varepsilon_0=\varepsilon_1=0.01$ and $\epsilon_0=\epsilon_1=0.01$. The r.v.s $Y_0,Y_1,\ldots,Y_n$, which are consistent with the observations $y_0,y_1,\ldots,y_n$, are i.i.d. The simulation is performed for six different distributions of $Y_1$ for each hypothesis. Under $\mathcal{H}_j$, $Y_1$ is distributed as one of the following distributions: the nominal distribution $F_i$ denoted by ($n$), LFD $\hat{Q}_j$ with parameters $\varepsilon_0=\varepsilon_1=0.01$ denoted by ($m$), LFD $\hat{Q}_i$ with parameters $\epsilon_0=\epsilon_1=0.01$ denoted by ($h$), LFD $\bar{G}_i$ of the asymptotically robust test with $\varepsilon_0=\varepsilon_1=0.01$ denoted by ($a$), LFD of the composite model $\hat{Q}_j$ with parameters $\varepsilon_0=\varepsilon_1=0.01$ and $\epsilon_0=\epsilon_1=0.01$ denoted by ($c$), $j\in\{0,1\}$. For comparison reasons, the sixth LFD is
introduced with respect to the composite uncertainty set. The LFD of the asymptotically robust test $\bar{G}_j$ for $\varepsilon_0=\varepsilon_1=0.01$
are first obtained. Then, $\hat{Q}_j$ with $\epsilon_0=\epsilon_1=0.01$ is determined when $\bar{G}_j$ is the nominal distribution, $j\in\{0,1\}$.
This test is denoted by $(c^*)$. Figure~\ref{fig5} and Fig.~\ref{fig6} illustrate $I_0$ and $I_1$ when $Y_1$ follows various distributions, as described above. The notation $|_a^b$ indicates that the robust test is performed by the LFDs of the (a)-test and the observations follow the LFD of (b)-test. In general, the composite test is not claimed to be asymptotically minimax robust since the LFDs of the (m)-test are not asymptotically robust. However, for this example, the (c)-test asymptotically does not degrade its performance for all observation models, when $t$ is small enough in its allowable limits. This test corresponds to the type I Neyman-Pearson test, cf.~\cite{levy}.
\begin{figure}[ttt]
  \centering
  \psfrag{x}[t][]{$\epsilon=\epsilon_0=\epsilon_1$}
  \psfrag{y}{$y_u$}
  \centerline{\includegraphics[width=8.5cm]{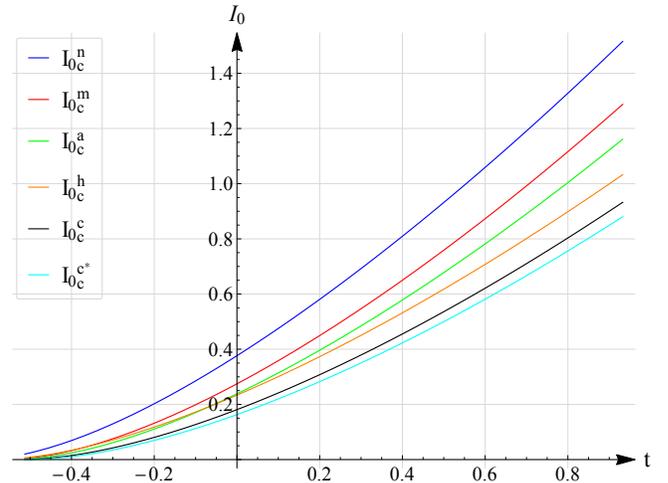}}
\vspace{-3mm}
\caption{Asymptotic decrease rate $I_0$ of the composite test when the observations follow the nominal distributions $F_0\sim\mathcal{N}(-1,1)$ and $F_1\sim\mathcal{N}(1,2)$, LFDs of the (m)-test, LFDs of the (a)-test, $\bar{g}_0$ and $\bar{g}_1$, with $\varepsilon_0=\varepsilon_1=0.01$, LFDs of the (h)-test, $\hat{q}_0$ and $\hat{q}_1$, with $\{\epsilon_0=\epsilon_1=0.01,\varepsilon_0=\varepsilon_1=0\}$, LFDs of the (c)-test, $\hat{q}_0$ and $\hat{q}_1$, with $\{\epsilon_0=\epsilon_1=0.01,\varepsilon_0=\varepsilon_1=0.01\}$ and LFDs of the ($c^*$)-test .\hspace{-3mm}}\label{fig5}
\end{figure}
\begin{figure}[ttt]
  \centering
  \psfrag{x}[t][]{$\epsilon=\epsilon_0=\epsilon_1$}
  \psfrag{y}{$y_u$}
  \centerline{\includegraphics[width=8.5cm]{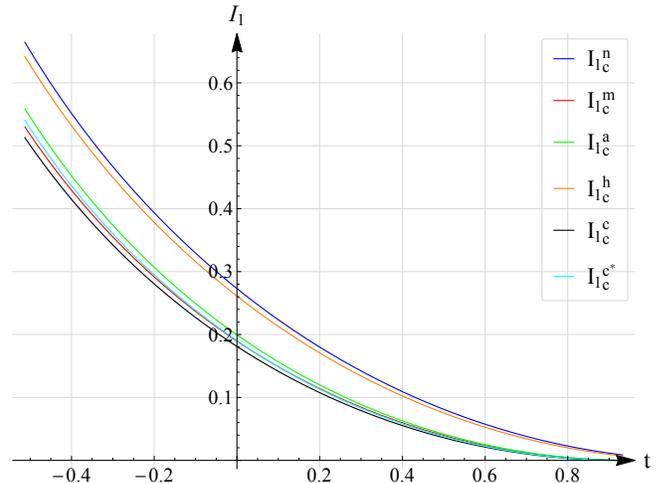}}
\vspace{-3mm}
\caption{Asymptotic decrease rate $I_1$ of the composite test when the observations follow nominal distributions $F_0\sim\mathcal{N}(-1,1)$ and $F_1\sim\mathcal{N}(1,2)$ LFDs of the (m)-test, LFDs of the (a)-test, $\bar{g}_0$ and $\bar{g}_1$, with $\varepsilon_0=\varepsilon_1=0.01$, LFDs of the (h)-test, $\hat{q}_0$ and $\hat{q}_1$, with $\{\epsilon_0=\epsilon_1=0.01,\varepsilon_0=\varepsilon_1=0\}$, LFDs of the (c)-test, $\hat{q}_0$ and $\hat{q}_1$, with $\{\epsilon_0=\epsilon_1=0.01,\varepsilon_0=\varepsilon_1=0.01\}$ and LFDs of the ($c^*$)-test .\hspace{-3mm}}\label{fig6}
\end{figure}\\
In the fifth simulation, a single sample (m)-test \eqref{equation25} is considered, when the nominal distributions are mean shifted and mean and variance shifted Gaussian distributions as defined before.
Robustness parameters are chosen to be equal ($\varepsilon=\varepsilon_0=\varepsilon_1$). For this choice, from \eqref{eq32}, it follows that $\varepsilon\in[0,0.5]$ for the mean shifted Gaussian distributions and $\varepsilon\in[0,0.338]$
for the mean and variance shifted Gaussian distributions s.t. the LFDs do not fully overlap. For all possible choices of $\varepsilon$, the performance of this robust test was calculated when the observations are due to LFDs of the (m)-test \eqref{equation24} and the LFDs of the (a)-test \eqref{equation61}, which are determined for the same $\varepsilon$ of the robust test.
The rationale behind this simulation is to test the minimax property defined by \eqref{eq9} and \eqref{equation12}.
The choice of the (a)-test as a competitor to the (m)-test is not arbitrary.
First, the LFDs of both tests lie on the boundary of the closed ball and second, the (a)-test is claimed to be asymptotically robust for large enough $n$ \cite{dabak}. Figure~\ref{fig7} illustrates the outcome of this simulation for the mean shifted Gaussian distributions. Due to the symmetry of the nominal distributions and the equal choice of the robustness parameters, we have $P_E=P_E^0=P_E^1$.
It can be seen that the robust test doesn't degrade its performance as expected. Similarly, in Fig.~\ref{fig8} the result of the same simulation for the mean and variance shifted Gaussian distributions is given. Since the nominal distributions are not symmetric, the error probabilities ($P_E^0$ and $P_E^1$) are unequal. More interestingly, as illustrated in Fig.~\ref{fig9}, the false alarm probability first increases with $\varepsilon$ and then starts decreasing. In all cases, it can be seen that \eqref{eq9} and \eqref{equation12} are valid.\\
The last part of the simulations is related to the robustness of the sequential probability ratio test based on the likelihood ratio between LFDs obtained by single sample robust tests. The robustness of the composite model strictly depends on the robustness of each single model: the sequential (m)-test and the sequential (h)-test. If one of them fails to be minimax robust, then the composite model is not minimax robust either. This makes the analysis of the test of robustness for the sequential (m)-test and the sequential (h)-test general enough to have conclusions about the composite test. In the sequel, Monte-Carlo simulations have been performed with $10^5$ samples. The threshold space $(\ln t_l,\ln t_u)\in\mathbb{R}^{-}\times\mathbb{R}^{+}$ is first cropped to $[-6,0]\times[0,6]$ and then discretized with a step parameter of $0.01$ in both directions, leading to $60\times 60$ pairs of $(\ln t_l,\ln t_u)$. The nominal distributions are selected to be the mean and variance shifted Gaussian
distributions as before. For $\varepsilon=\varepsilon_0=\varepsilon_1=0.01$, the LFDs of the (m)-test $(\hat{g}_0,\hat{g}_1)$ and the (a)-test $(\bar{g}_0,\bar{g}_1)$ are determined by solving \eqref{equation29}, \eqref{equation30} and \eqref{equation62}. Accordingly, the likelihood ratio is formed by $\hat{l}=\hat{g}_1/\hat{g}_0$ or $\bar{l}=\bar{g}_1/\bar{g}_0$. The tests considered are $\hat{S}_n=\sum_{i=1}^n \hat{l}(Y_i)$ and $\bar{S}_n=\sum_{i=1}^n \bar{l}(Y_i)$ where every $Y_i$ is distributed either as $\hat{g}_0$ or $\bar{g}_0$ under $\mathcal{H}_0$ and either $\hat{g}_1$ or $\bar{g}_1$ under $\mathcal{H}_1$. For every pair of thresholds $(\ln t_l,\ln t_u)$, the sequential test is run and the false alarm probability, miss detection probability and expected number of samples under $\mathcal{H}_0$ and $\mathcal{H}_1$ are calculated. Figure~\ref{fig10} illustrates the ratio of the false alarm probability $\alpha_m^a$ to the false alarm probability $\alpha_m^m$. Clearly, the
performance of $\hat{S}_n$ for $Y_i \sim \hat{g}_0$ degrades for almost all simulation points if actually $Y_i \sim \bar{g}_0$. Figure~\ref{fig11} illustrates similar results for the miss detection probability when the robust test is $\bar{S}_n=\sum_{i=1}^n \bar{l}(Y_i)$. Again, the test doesn't satisfy the bounded error probability condition. Figures~\ref{fig12}-\ref{fig15} illustrate the same type of simulations for the expected number of samples where similar observations can be made. In conclusion, one can see that the sequential (m)-test is not robust for the error probability as well as for the expected number of samples, whereas the sequential (a)-test is only asymptotically robust for the expected number of samples. The simulation results are in agreement with the theoretical findings. A short comparison of the (m)-test, the (a)-test and the (h)-test are given in Table~\ref{tab1}.

\begin{figure}[ttt]
  \centering
  \psfrag{x}[t][]{$\epsilon=\epsilon_0=\epsilon_1$}
  \psfrag{y}{$y_u$}
  \centerline{\includegraphics[width=8.5cm]{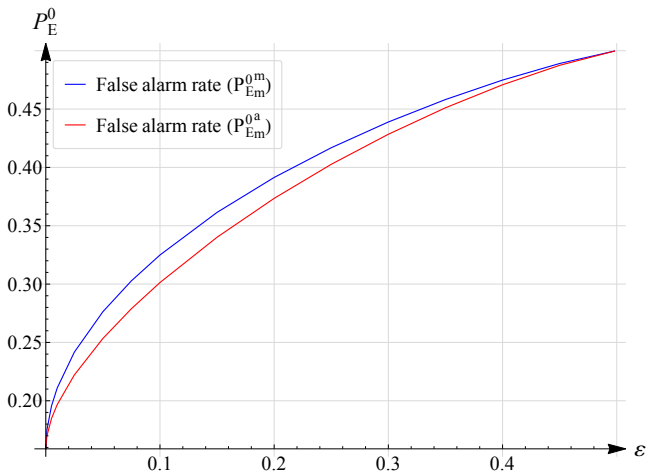}}
\vspace{-3mm}
\caption{The performance of the robust decision rule \eqref{equation25} for all achievable least favorable distributions of the (m)-test and the (a)-test with $\varepsilon=\varepsilon_0=\varepsilon_1$, when the nominal distributions are $F_0\sim\mathcal{N}(-1,1)$ and $F_1\sim\mathcal{N}(1,1)$.\hspace{-3mm}}\label{fig7}\hspace{-2mm}
\vspace{-3mm}
\end{figure}

\begin{figure}[ttt]
  \centering
  \psfrag{x}[t][]{$\epsilon=\epsilon_0=\epsilon_1$}
  \psfrag{y}{$y_u$}
  \centerline{\includegraphics[width=8.5cm]{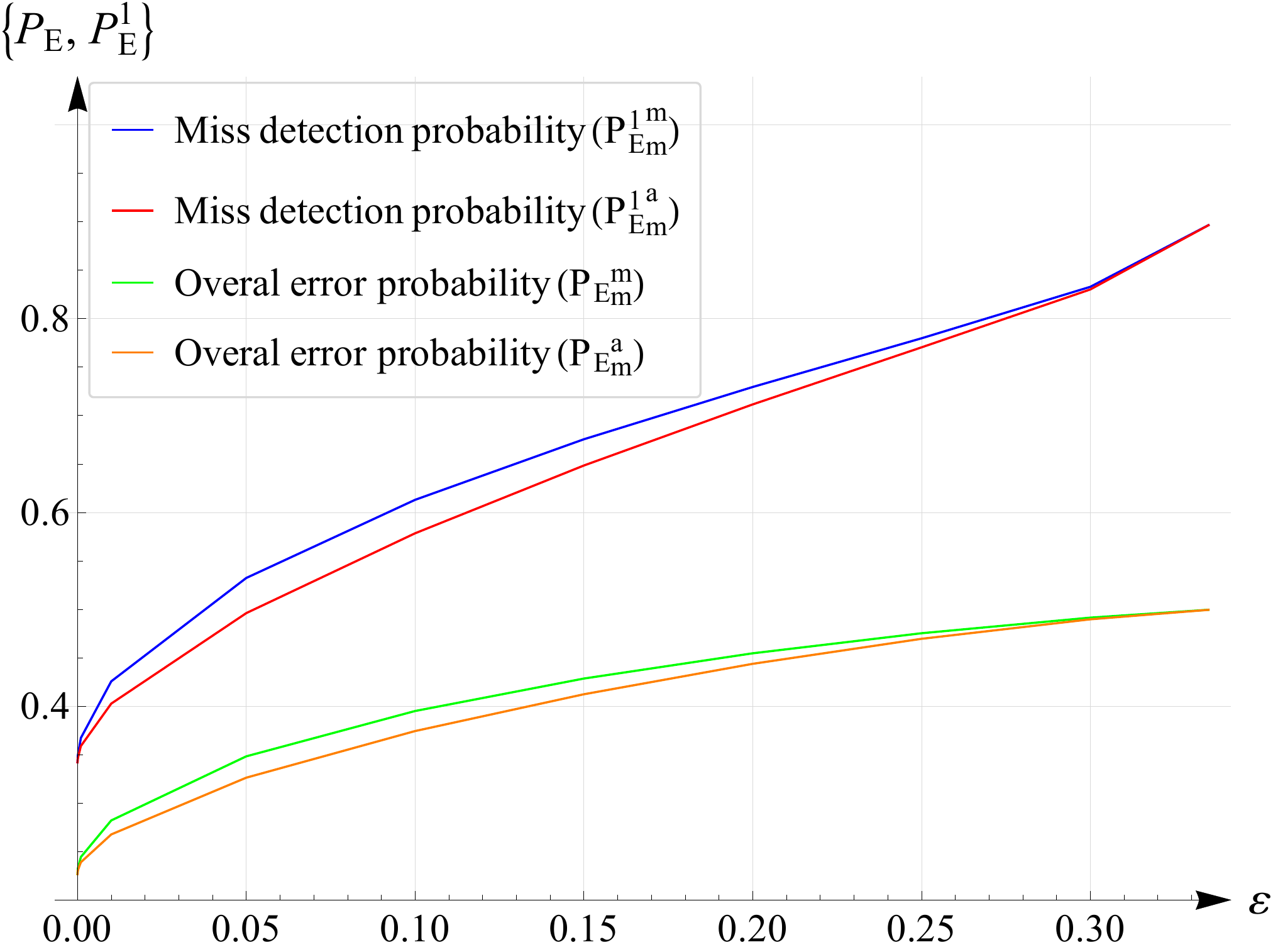}}
\vspace{-3mm}
\caption{The performance of the robust decision rule \eqref{equation25} for all achievable least favorable distributions of the (m)-test and the (a)-test with $\varepsilon=\varepsilon_0=\varepsilon_1$, when the nominal distributions are $F_0\sim\mathcal{N}(-1,1)$ and $F_1\sim\mathcal{N}(1,2)$.\hspace{-3mm}}\label{fig8}\hspace{-2mm}
\vspace{-3mm}
\end{figure}

\begin{figure}[ttt]
  \centering
  \psfrag{x}[t][]{$\epsilon=\epsilon_0=\epsilon_1$}
  \psfrag{y}{$y_u$}
  \centerline{\includegraphics[width=8.5cm]{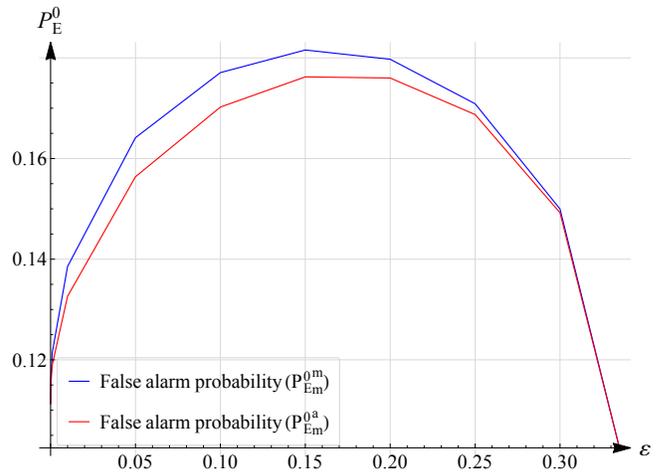}}
\vspace{-3mm}
\caption{False alarm probability of the robust decision rule \eqref{equation25} for all achievable least favorable distributions of the (m)-test and the (a)-test with $\varepsilon=\varepsilon_0=\varepsilon_1$, when the nominal distributions are $F_0\sim\mathcal{N}(-1,1)$ and $F_1\sim\mathcal{N}(1,2)$.\hspace{-3mm}}\label{fig9}\hspace{-2mm}
\vspace{-3mm}
\end{figure}

%\begin{figure}[ttt]
%  \centering
%  \psfrag{x}[t][]{$\epsilon=\epsilon_0=\epsilon_1$}
%  \psfrag{y}{$y_u$}
%  \centerline{\includegraphics[width=8.5cm]{figures/loglikelihoodpdfs.pdf}}
%\vspace{-3mm}
%\caption{blabla.\hspace{-3mm}}\label{fig9}\hspace{-2mm}
%\vspace{-3mm}
%\end{figure}

\begin{figure}[thb]
  \centering
  \psfrag{x}[t][]{$\epsilon=\epsilon_0=\epsilon_1$}
  \psfrag{y}{$y_u$}
  \centerline{\includegraphics[width=8.5cm]{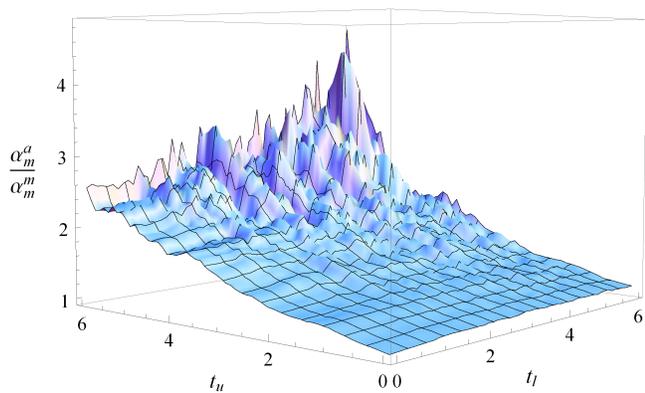}}
\vspace{-3mm}
\caption{The ratio of the false alarm probabilities of the (m)-test when the observations follow the LFD of the (a)-test ($\hat{G}_0$) and the LFD of the (m)-test ($\bar{G}_0$).\hspace{-3mm}}\label{fig10}\hspace{-2mm}
\vspace{-3mm}
\end{figure}

\begin{figure}[thb]
  \centering
  \psfrag{x}[t][]{$\epsilon=\epsilon_0=\epsilon_1$}
  \psfrag{y}{$y_u$}
  \centerline{\includegraphics[width=8.5cm]{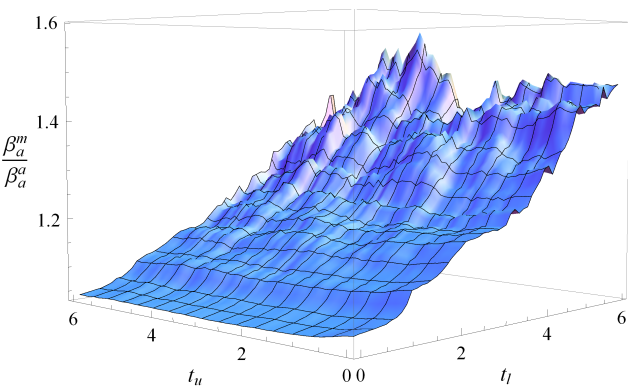}}
\vspace{-3mm}
\caption{The ratio of the miss detection probabilities of the (a)-test when the observations follow the LFD of the (m)-test ($\hat{G}_0$) and the LFD of the (a)-test ($\bar{G}_0$).\hspace{-3mm}}\label{fig11}\hspace{-2mm}
\vspace{-3mm}
\end{figure}

\begin{figure}[thb]
  \centering
  \psfrag{x}[t][]{$\epsilon=\epsilon_0=\epsilon_1$}
  \psfrag{y}{$y_u$}
  \centerline{\includegraphics[width=8.5cm]{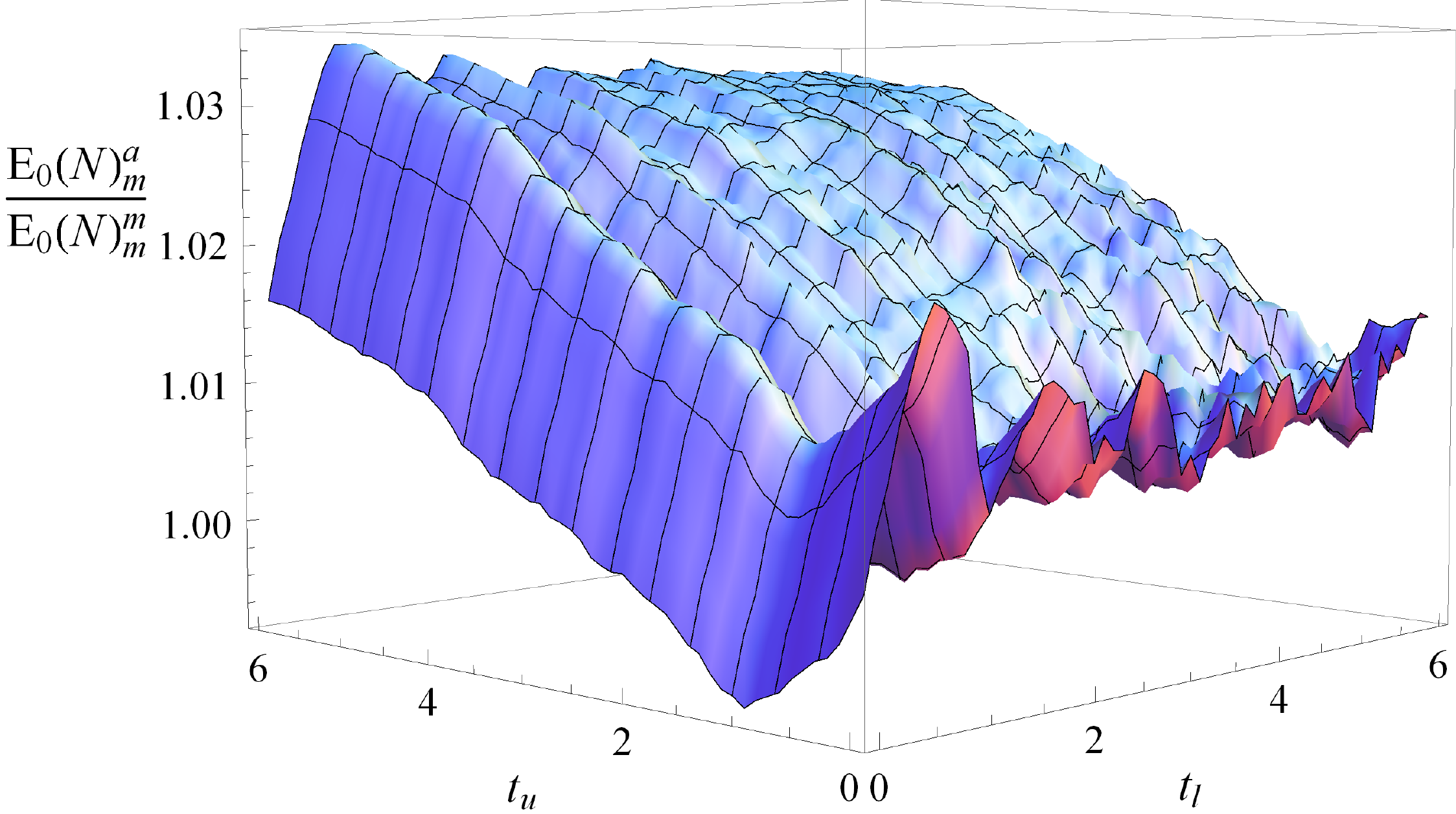}}
\vspace{-3mm}
\caption{The ratio of the expected number of samples of the (m)-test when the observations follow the LFD of the (a)-test ($\hat{G}_0$) and the LFD of the (m)-test ($\bar{G}_0$).\hspace{-3mm}}\label{fig12}\hspace{-2mm}
\vspace{-3mm}
\end{figure}

\begin{figure}[thb]
  \centering
  \psfrag{x}[t][]{$\epsilon=\epsilon_0=\epsilon_1$}
  \psfrag{y}{$y_u$}
  \centerline{\includegraphics[width=8.5cm]{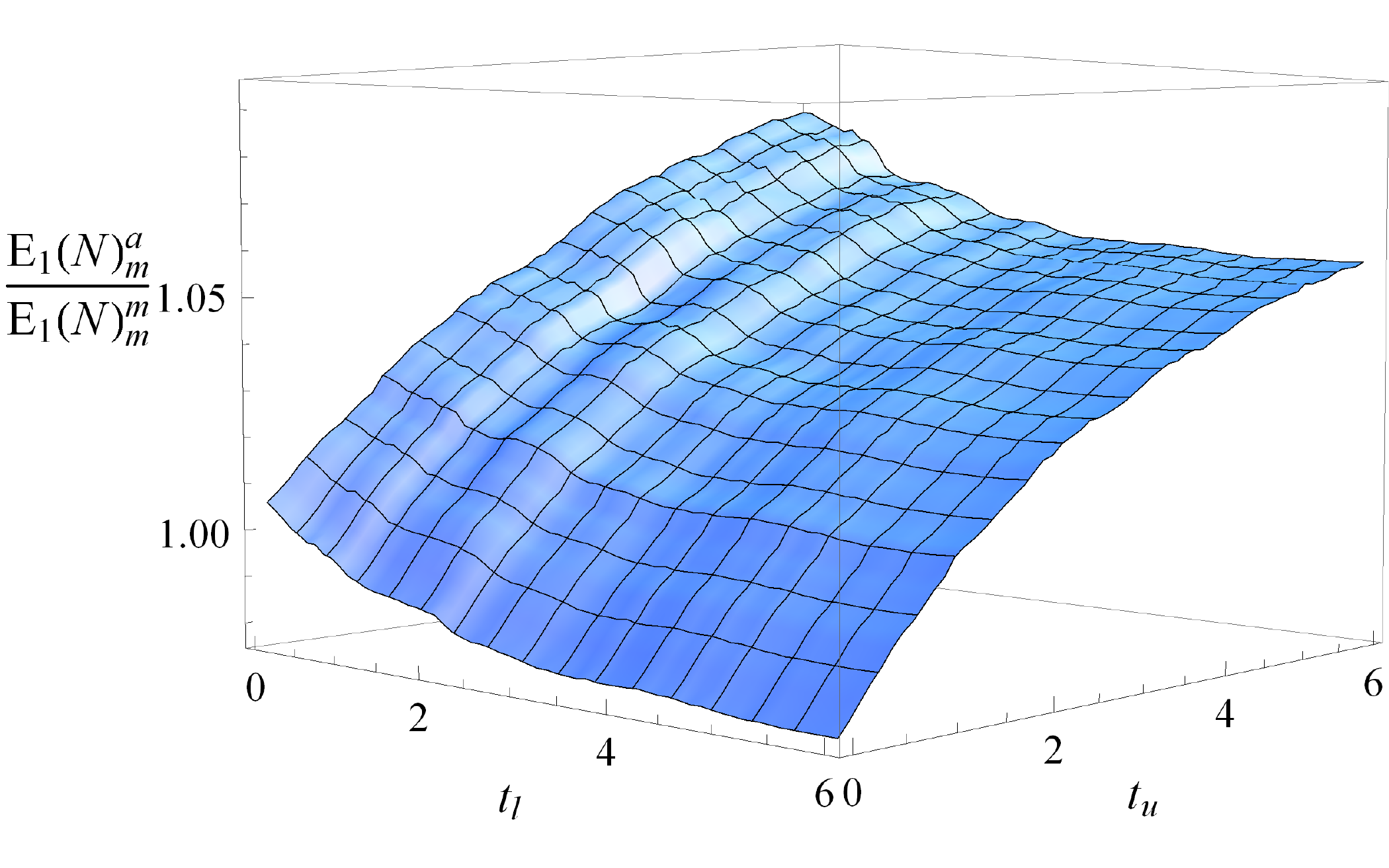}}
\vspace{-3mm}
\caption{The ratio of the expected number of samples of the (m)-test when the observations follow the LFD of the (a)-test ($\hat{G}_1$) and the LFD of the (m)-test ($\bar{G}_1$).\hspace{-3mm}}\label{fig13}\hspace{-2mm}
\vspace{-3mm}
\end{figure}

\begin{figure}[thb]
  \centering
  \psfrag{x}[t][]{$\epsilon=\epsilon_0=\epsilon_1$}
  \psfrag{y}{$y_u$}
  \centerline{\includegraphics[width=8.5cm]{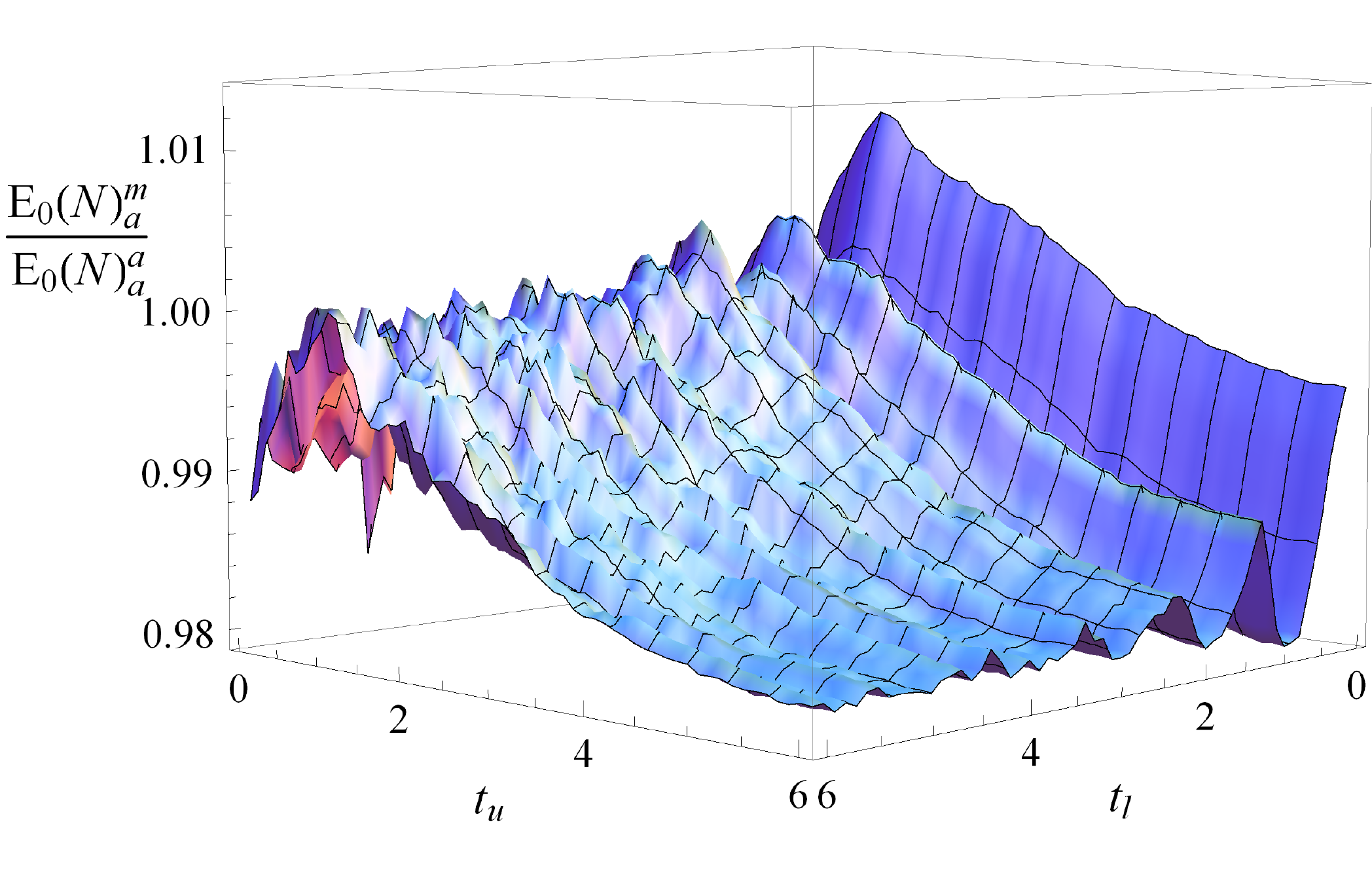}}
\vspace{-3mm}
\caption{The ratio of the expected number of samples of the (a)-test when the observations follow the LFD of the (m)-test ($\hat{G}_0$) and the LFD of the (a)-test ($\bar{G}_0$).\hspace{-3mm}}\label{fig14}\hspace{-2mm}
\vspace{-3mm}
\end{figure}

\begin{figure}[thb]
  \centering
  \psfrag{x}[t][]{$\epsilon=\epsilon_0=\epsilon_1$}
  \psfrag{y}{$y_u$}
  \centerline{\includegraphics[width=8.5cm]{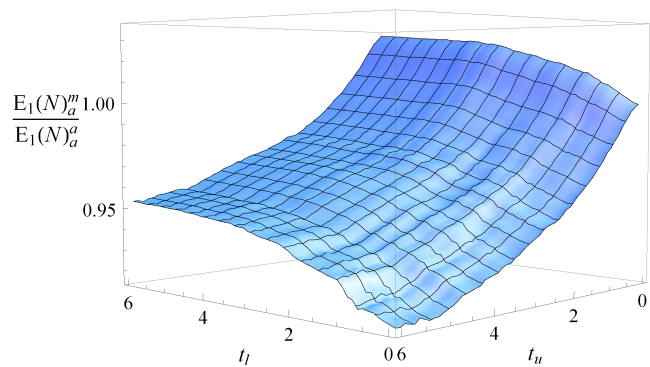}}
\vspace{-3mm}
\caption{The ratio of the expected number of samples of the (a)-test when the observations follow the LFD of the (m)-test ($\hat{G}_1$) and the LFD of the (a)-test ($\bar{G}_1$).\hspace{-3mm}}\label{fig15}\hspace{-2mm}
\vspace{-3mm}
\end{figure}

\section{Conclusion}\label{conclusions}
A minimax robust hypothesis testing scheme between two composite hypotheses based on the KL-divergence has been proposed. It has been shown that the proposed model reduces to Levy's robust test \cite{levy09} when the nominal likelihood ratio is monotone and the nominal probability density functions are symmetric. For comparison purposes, Dabak's asymptotically robust test \cite{dabak} has been introduced and the existence of LFDs for this test has been proven without consideration of the geometrical aspects of hypothesis testing. It has been shown that the proposed minimax robust test, the (m)-test, can be combined with Huber's clipped likelihood ratio test, the (h)-test, in a composite uncertainty model. Hence, the composite test, the (c)-test, provides minimax robustness both for outliers as well as for modeling errors. The existence of LFDs for the composite uncertainty model has also been proven. It has been demonstrated that the proposed composite model reduces to the individual robust tests via a suitable choice of the parameters.\\
To design a robust test for modeling errors, the uncertainty sets can be constructed by choosing distances different from the KL-divergence. It has been shown that the choice of a distance plays a crucial role in designing the robust tests. Although the robust version of the likelihood ratio test remains the same for many distances, there are examples where this assertion is not true.
Among several distances discussed, the symmetrized $\chi^2$ has been found to be more suitable
for the design of a robust hypothesis test if the tail structures of the nominal distributions are needed to be roughly preserved. It has been also shown that the maximum robustness parameters are bounded from above. Both for the (m)-test as well as for the (h)-test, the problem of determining the maximum robustness parameters is proven to be a convex optimization problem, and therefore the related equations can be solved by a polynomial time algorithm.\\
Next, the single sample robust tests have been extended to fixed sample size tests. Cram\'{e}r's theorem has been adopted to characterize the asymptotic behavior of the robust tests. Interestingly, it has been found that the formulation of the asymptotic decrease rate of the error probability for the fixed sample size test coincides with the formulation to determine the maximum robustness parameters for the (m)-test. Later, single sample robust tests have been extended to the sequential hypothesis test. The minimax properties of the considered robust tests have either been proven or disproven analytically or with simulations. Finally, we have justified that the proposed composite model is applicable for robust estimation problems. Various simulation results show the agreement with theoretical findings.

\begin{table}[ht]
\caption{Comparison between the robust tests}
\begin{center}
\begin{tabular}{|l|l|l|l|}
\hline
&(m)-test & (a)-test & (h)-test \\
\hline \hline
Unique LFDs & Yes & Yes &No  \cite{hube65}\\
\hline
Unique test & Yes & Yes &Yes\\
\hline
Limiting test &  Soft sign test & Like. ratio test & Sign test  \\
\hline
Suitable for &  Model. errors & Model. errors & Outliers  \\
\hline
Non-linear equations &  Two coupled & Two distinct & Two distinct \\
\hline
Number of samples &  $n=1$ & $n\rightarrow\infty$ & $1<n<\infty$  \\
\hline
Fixed sample size test  &Not robust & Asymp. rob. \cite{dabak}& Robust  \\
\hline
Sequential test, $(\alpha,\beta)$  &Not robust & Not robust& Robust \\
\hline
Sequential test, $E[N]$  &Not robust & Asymp. rob.& Asymp. rob.\\
\hline
\end{tabular}
\end{center}
\label{tab1}
\end{table}

%\begin{table}[ht]
%\caption{Performance Degradation of The Composite model}
%\begin{center}
%\begin{tabular}{|l|l|l|l|l|l|}
%\hline
%&Nominal&Mod. Errors & Asymp. Rob. & Outliers \\
%\hline \hline
%False alarm&$0$& $0$ & $0$ & $0$  \\
%\hline
%Miss detection&$0$& $0.05$ & $0$ & $0$  \\
%\hline
%$E_0(N)$&$0.0044$& $0.3006$ & $0.4269$ & $0$    \\
%\hline
%$E_1(N)$&$0.0253$& $0.1897$ & $0.4397$ & $0.0181$   \\
%\hline
%\end{tabular}
%\end{center}
%\label{my_table1}
%\end{table}

\section*{Acknowledgment}
\label{sec:acknowledge}
This work was supported by the LOEWE Priority Program Cocoon (http://www.cocoon.tu-darmstadt.de).

\bibliographystyle{IEEEtran}
\bibliography{strings3}
\end{document}